\documentclass[11pt,letterpaper]{article}
\usepackage{amsmath,amsfonts,amsthm,amssymb}
\usepackage{setspace}
\usepackage{fancyhdr}
\usepackage{lastpage}
\usepackage{extramarks}
\usepackage{xspace}

\usepackage{chngpage}
\usepackage{algorithm,algorithmicx}
\usepackage[noend]{algpseudocode}
\usepackage{soul,color}
\usepackage{graphicx,float,wrapfig}
\usepackage[font=small,labelfont=bf]{caption}
\usepackage[margin=1in]{geometry}
\usepackage{thmtools,thm-restate}
\usepackage{verbatim}

\usepackage[nice]{nicefrac}

\usepackage{mathrsfs}

\allowdisplaybreaks

\usepackage{mdframed}
\newtheorem*{mdresult}{Result}

\newcommand{\R}{\mathbb{R}}

\newcommand{\lp}{\mathrm{LP}}

\newcommand{\epseh}{\epsilon_{\textsc{E}}}
\newcommand{\freasy}{\mathsf{FR}_{\textsc{E}}}
\newcommand{\we}{W_{\textsc{E}}}
\newcommand{\frhard}{\mathsf{FR}_{\textsc{L}}}

\newcommand{\fr}{\mathsf{FR}}

\newcommand{\dt}{\mathsf{DT}}
\newcommand{\opt}{\mathsf{Opt}}

\newcommand{\alg}{\mathsf{alg}}
\newcommand{\req}{\mathsf{req}}
\newcommand{\get}{\mathsf{get}}

\newcommand{\D}{\mathcal{D}}
\newcommand{\sipi}{\mathcal{PI}}
\newcommand{\wesipi}{\mathcal{PI}}
\newcommand{\bssipi}{\mathcal{PI}^{\mathsf{bl}}}

\newcommand{\bm}{\mathsf{Benchmark}}

\newcommand{\E}{\mathbb{E}}

\newcommand{\one}{\mathbf{1}}%

\newcommand{\pr}{\mathbf{Pr}} 

\newcommand{\lpon}{{\mathsf{LP}_{\textsc{on}}}}
\newcommand{\lpdual}{{\mathsf{LP}_{\textsc{dual}}}}

\newcommand{\opton}{{\mathsf{OPT}_{\textsc{on}}}}

\usepackage[usenames,dvipsnames]{xcolor}
\usepackage[linktocpage=true,breaklinks,colorlinks,citecolor=blue,linkcolor=BrickRed]{hyperref}
\usepackage{cleveref}
\newcommand{\IGNORE}[1]{}
\newcommand{\poly}{\ensuremath{\mathsf{poly}}}

\usepackage{tikz}
\usetikzlibrary{graphs, graphs.standard}
\usetikzlibrary{arrows.meta}

\usepackage{tcolorbox}
\usepackage{enumerate}

\title{Combinatorial Philosopher Inequalities}

 \author{
 Enze Sun\thanks{The University of Hong Kong (\texttt{sunenze@connect.hku.hk})}
 \and
 Zhihao Gavin Tang\thanks{ITCS, Key Laboratory of Interdisciplinary Research of Computation and Economics, Shanghai University of Finance and Economics (\texttt{tang.zhihao@mail.shufe.edu.cn})}
 \and
 Yifan Wang\thanks{School of Computer Science, Georgia Tech (\texttt{ywang3782@gatech.edu})}
 }
\date{}

\newtheorem{Theorem}{Theorem}[section]
\newtheorem*{Theorem*}{Theorem}
\newtheorem{Lemma}[Theorem]{Lemma}

\newtheorem{Observation}[Theorem]{Observation}

\newtheorem{Claim}[Theorem]{Claim}

\newtheorem*{Lemma*}{Lemma}
\newtheorem{Example}[Theorem]{Example}

\begin{document}

\maketitle 

\begin{abstract}
In online combinatorial allocation, agents arrive sequentially and items are allocated in an online manner. The algorithm designer only knows the distribution of each agent's valuation, while the actual realization of the valuation is revealed only upon her arrival. 
Against the offline benchmark, Feldman, Gravin, and Lucier (SODA 2015) designed an optimal $0.5$-competitive algorithm for XOS agents. An emerging line of work focuses on designing approximation algorithms against the (computationally unbounded) optimal online algorithm. The primary goal is to design algorithms with approximation ratios strictly greater than $0.5$, surpassing the impossibility result against the offline optimum. Positive results are established for unit-demand agents (Papadimitriou, Pollner, Saberi, Wajc, MOR 2024), and for $k$-demand agents (Braun, Kesselheim, Pollner, Saberi, EC 2024).

In this paper, we extend the existing positive results for agents with submodular valuations by establishing a $0.5 + \Omega(1)$ approximation against a newly constructed online configuration LP relaxation for the combinatorial allocation setting. Meanwhile, we provide negative results for agents with XOS valuations by providing a $0.5$ integrality gap for the online configuration LP, showing an obstacle of existing approaches.
\end{abstract}

\section{Introduction}

We study the online combinatorial allocation problem with  $m$ items and $T$ stochastic agents. Each agent $t \in [T]$\footnote{For every integer $n$, we use $[n]$ to denote the set of integers $\{1,2,...,n\}$.} is associated with an individual valuation function $v_t: 2^{[m]} \to \R_{\geq 0}$, where $v_t$ is independently drawn from a priori known distribution $\D_t$. 
The agents arrive in a sequence. Upon the arrival of the $t$-th agent, her individual valuation function $v_t$ is revealed to the algorithm, and the algorithm immediately allocates a subset $S_t$ of currently available items to the agent. 
Each item can be allocated at most once, i.e., $\{S_t\}$ are mutually disjoint. The goal is to maximize the total expected reward $\E[\sum_{t} v_t(S_t)]$.
The problem is also known as the online combinatorial auction. As we are not concerned with the incentive issues for truthfully reporting $v_t$ to the algorithm, we choose to use the term of allocation. 

To evaluate the performance of an algorithm, the traditional wisdom is to use the \emph{offline optimum} as the benchmark. These results are conventionally referred as \emph{prophet inequalities}, since the benchmark essentially corresponds to the optimal allocation when the algorithm sees every valuation function $v_t$ in advance. And this omniscient power is only mastered by \emph{prophet}.
Optimal prophet inequalities have been established in various settings.  
Krengel and Sucheston~\cite{krengel1978}, Samuel-Cahn~\cite{Samuel-Annals84} established an optimal $0.5$ prophet inequality in the single-choice setting, i.e., when $m=1$. Feldman, Gravin, and Lucier~\cite{DBLP:journals/corr/FeldmanGL14} attained the same $0.5$ ratio in the setting of XOS valuations. Very recently, Correa and Cristi~\cite{DBLP:conf/stoc/CorreaC23} established a constant prophet inequality for subadditive valuations. 

An emerging line of works uses the \emph{online optimum} as the benchmark for online combinatorial allocation. Noticeably, given the distributions $\{\D_t\}$ and the arrival order of the agents, the concept of the optimal online algorithm is well-defined. The natural research question is to design efficient poly-time approximation algorithms against the online optimum. These results are referred as \emph{philosopher inequalities}, where philosopher refers to the computationally unbounded optimal algorithm.
The primary goal is to establish an approximation ratio strictly greater than the competitive ratio from prophet inequalities.
Positive results are known for unit-demand agents (a.k.a., the matching setting)~\cite{PPSW-MOR24,BDL-EC22,NSW-SODA25,BDPSW-SODA25}, for $k$-demand agents (a.k.a., the capacitated setting)~\cite{BKPS-EC24}, and for stationary matching~\cite{AAPS-arxiv24}.

\subsection{Our Results}
In this work, we develop philosopher inequalities for online combinatorial allocation with submodular valuations, generalizing the existing positive results to a richer family of valuation functions. Meanwhile, we provide negative results for XOS valuations, showing an obstacle of existing approaches (including ours). 

Since Papadimitriou et al.~\cite{PPSW-MOR24} established a hardness result for the unit-demand setting, all approximation algorithms start with a tractable relaxation of the online optimum. To the best of our knowledge, existing works rely on the same \emph{online constraint} in addition to the standard ex-ante relaxation. Intuitively, the online constraint captures the fact that decisions made before time $t$ are independent of the realization of the $t$-th agent. Therefore, the probability that an item $i$ is allocated to the $t$-th agent with a certain type is bounded by the probability that this type is realized times the probability that item $i$ is allocated before $t$'s arrival.
We follow the same strategy and construct an \emph{online configuration LP}, which relaxes the reward achieved by the optimal online algorithm, resolving the inherit difficulty of combinatorial valuation functions. 

% We refer to it as the online configuration LP, and the constraint~\eqref{eq:online-constraint} as the online constraint. Although the program includes an exponential number of parameters, it is efficiently solvable assuming \emph{demand oracles}.
Our main results are the following: on the positive side, we show that online combinatorial allocation admits a $0.5 + \Omega(1)$ philosopher inequality, assuming agents have submodular valuations:

\begin{Theorem*}[Informal \Cref{thm:lp-main}]
For online combinatorial allocation problem with submodular agents, there exists a poly-time $0.5+\Omega(1)$ approximation algorithm against the online configuration LP, assuming demand oracles.
\end{Theorem*}

On the negative side, we show that our result for submodular agents does not carry over for the class of XOS functions:

\begin{Theorem*}[Informal \Cref{thm:xos-hard}]
For online combinatorial allocation problem with XOS agents, the integrality gap of the online configuration LP is $0.5$.
\end{Theorem*}
We find the separation between submodular valuations and XOS valuations unexpected. Indeed, without the online constraint, online combinatorial allocation admits the same $0.5$ competitive ratio, for both submodular and XOS agents~\cite{DBLP:journals/corr/FeldmanGL14}. Our finding calls for novel relaxation of the online optimum, for generalizing philosopher inequalities beyond submodular valuations.

\subsection{Our Techniques}

{

Deviating from the existing approaches in the philosopher inequalities which mainly relies on bounding the positive correlation introduced in the rounding procedure (see, e.g., \cite{PPSW-MOR24, BKPS-EC24}), our positive result relies on a novel rounding algorithm for the instances with ``free-deterministic'' structure. To be specific, our algorithm proceeds as the following after solving the online configuration LP.

\paragraph{Easy instance check.} We first perform an ``easy instance check'' for the optimal solution of the online configuration LP. This step determines whether the instance admits an approximation strictly better than $0.5$ using standard prophet inequality algorithms. 
%, which are at least $0.5$ approximation. 
Specifically, we apply a standard reduction from the submodular prophet inequality to a collection of single-item prophet inequality instances. This reduction relies on the fact that every submodular function belongs to the class of XOS functions. Consequently, the value of a submodular function $ v(S) $ can be expressed as $ v(S) = \sum_{j \in S} a_j $ for some additive values $ a_j $. By assigning each value $ a_j $ to the prophet inequality instance corresponding to item $ j $, we ensure that the algorithm achieves at least $0.5$ times the objective value of the online configuration LP.

\paragraph{Free-deterministic structure.}  After performing the easy instance check, we are done if the standard prophet inequality algorithm achieves an approximation ratio strictly better than $0.5$. Otherwise, we argue each single-item prophet inequality instance constructed in the above reduction must be hard, exhibiting a \emph{free-deterministic} structure. The \emph{free-deterministic decomposition}, introduced by~\cite{stoc/STW25}, captures the structure of such hard instances. This is motivated by the classic lower-bound construction for single-item prophet inequality: a \textbf{deterministic} value of 1 arrives first, followed by a large value of $\delta^{-1}$, but with a small probability $ \delta \to 0$ being realized. We refer to this value as \textbf{free} value, as it can be taken for free whenever it arrives, because in expectation it consumes nearly no probability mass. If a single-item prophet inequality instance fails to admit a $(0.5 + \epsilon)$-competitive algorithm, the free-deterministic decomposition states that the values in the instance can be partitioned into a \emph{free} part and a \emph{deterministic} part. Each part contributes roughly half to the optimal offline reward, but they differ significantly in probability mass. The free part, corresponding to the large value $\delta^{-1} $ that arrives with tiny probability $\delta$, has nearly zero probability mass. In contrast, the deterministic part, corresponding to the value $1$ that arrives with probability $1$, carries nearly $1$ of the probability mass. 

When the given instance does not pass the easy instance check, we apply the free-deterministic decomposition. Thus, it remains to design an algorithm that achieves strictly greater than $0.5$ approximation ratio, assuming the instance follows the free-deterministic structure.

\paragraph{Core algorithm: half-double sampling.} The main technical contribution of our paper is the introduction of a novel rounding algorithm for instances with a free-deterministic structure. The motivation for our algorithm comes from a key property of this structure: since the free part contributes approximately $0.5$ to the total reward but has nearly zero probability mass, simply assigning the free part to each agent yields a $0.5$-approximation, while the probability that any item remains unallocated until the end of the algorithm is nearly $1$. Notably, this property still holds even if we attempt to allocate the free part more than once.

This observation inspires the idea of \textbf{``double sampling''}: when each agent arrives, we independently sample twice from the distribution defined by the solution of the online configuration LP. This produces two independent subsets that the agent is interested in, and we allocate to the agent the union of the free parts of these two subsets.

The above algorithm guarantees that each item is allocated with probability at most twice the probability mass of the free part (which is still nearly zero), while still achieving at least a $0.5$-approximation (since the  gained reward is no worse than that of sampling once). At a first glance, one might suppose that the algorithm achieves a better-than-$0.5$ approximation, as the idea of double sampling appears to double our chance of obtaining a high value (from the free part). 

Unfortunately, this algorithm fails to achieve a better-than-$0.5$ approximation ratio even when agents are unit-demand. See \Cref{sec:motivate} for a detailed discussion. However, we find the idea surprisingly powerful when applied only to \textbf{half} of the items, which leads to our half-double sampling algorithm.

Specifically, we sample a subset $Q \subseteq [m]$ of items by including each item independently with probability $0.5$. For items in $Q$, we apply the double sampling strategy described above, allocating an item to an agent if it appears in the union of the free parts of two independently sampled subsets. For items in $[m] \setminus Q$, we treat this portion as a new online combinatorial allocation instance and apply the standard prophet inequality algorithm. Our analysis shows that the above algorithm achieves $0.5625$ times the online configuration LP.

We remark that our algorithm shares some similarities to the algorithm in \cite{PPSW-MOR24}, which also uses the idea of ``sampling twice''. However, the underlying principles of two algorithms are different. The algorithm in \cite{PPSW-MOR24} samples a second time only when the first sample fails to form a matching edge, while in our algorithm, there is no difference between two samples and we sample a second time simply because it is almost for free.

}

\paragraph{Roadmap.}
Section~\ref{sec:prelim} introduces the notations, the online configuration LP, and the necessary preliminaries regarding single-item prophet inequality and submodular functions. The formal arguments for the easy instance check algorithm and the free-deterministic decomposition are given in Section~\ref{sec:submod}. We also give a more detailed motivation of our half-double sampling algorithm in \Cref{sec:motivate}.
Section~\ref{sec:wesmall} is dedicated to formalize the half-double sampling algorithm. Finally, the negative result for XOS valuations is given in Section~\ref{sec:xos}.

\subsection{Further Related Works}

\paragraph{Combinatorial Prophet Inequalities.} The stochastic online combinatorial allocation problem 
is originated from the classic single-choice online stopping problem, a.k.a., the Prophet Inequality, whose history can be traced back to the works of \cite{Krengel-Journal77,Krengel-Journal78,Samuel-Annals84}. The Prophet Inequality model aims at designing competitive algorithms against the expected offline optimum. The model is widely used in mechanism design, with extensions to matroids \cite{DBLP:journals/geb/KleinbergW19}, matchings \cite{EFGT-EC20}, submodular/XOS functions \cite{DBLP:journals/corr/FeldmanGL14, DFKL-FOCS17}, subadditive functions \cite{DBLP:conf/focs/DuttingKL20, DBLP:conf/stoc/CorreaC23}, and even general monotone functions \cite{GSW-STOC25}. 
For a comprehensive overview, we refer to the survey \cite{Lucier-Survey17}.

\paragraph{Online Optimum Benchmark.}
Besides the previously mentioned works, Aouad and Sarita{\c{c}}~\cite{or/AouadS22}, and AmaniHamedani et al.~\cite{AAPS-arxiv24} studied the stationary bipartite matching against the online optimum.
The concept of order competitive ratio is proposed by Ezra et al.~\cite{EFGT-SODA23}, to measure the significance of knowing the arrival order in online stochastic settings. For the single-choice setting, they designed an optimal deterministic order-uaware algorithm against the online benchmark. Chen, Sun, and Tang~\cite{CST-EC24} further studied randomized algorithms. Ezra and Garbuz~\cite{wine/EzraG23} applied the framework to combinatorial settings. Motivated by fairness concerns, Ezra, Feldman, and Tang~\cite{ec/EzraFT24} studied a more restricted set of anonymous algorithms for competing against the online optimum. % Stochastic matching against online optimum: \cite{SW-icalp21,BDL-EC22,NSW-SODA25,BDPSW-SODA25} 

While the works discussed above aim to achieve a constant-factor approximation against the online optimum, a different line of research seeks to achieve a $1-\epsilon$ approximation via a polynomial-time approximation scheme (PTAS). Positive results along this line include the free-order Prophet Inequality problem \cite{SS-EC21, LLPSS-EC21}; the Prophet Secretary problem \cite{DGRTT-EC23}; and the Pandora’s Box problem with nonobligatory inspection \cite{FLL-STOC23, BC-STOC23}.

% Unknown order single-item: \cite{EFGT-SODA23, CST-EC24}

% Unknown order matching:\cite{STW-arXiv25}

% Capacitated Resource Allocation (covered by submodular): \cite{DBLP:conf/sigecom/BraunKPS24}.

% Stationary Bipartite Matching against online optimum: \cite{AAPS-arxiv24}

% Comparing to online optimum (additively) via learning:
% \cite{GHTZ-COLT21, GKSW-SODA24}

\section{Preliminaries}
\label{sec:prelim}

\subsection{Valuation Functions}
We first introduce the notation and specify our computational model. We assume that each distribution $\D_t$ has support size $K$, and the valuation function $v_t$ equals $v_{t,k}: 2^{[m]} \to \R_{\geq 0}$ with probability $p_{t,k}$, for every $k \in [K]$. Naturally, $\sum_{k \in [K]} p_{t,k} = 1$. We assume that the valuation functions are nonnegative and monotone, i.e., for every $A \subseteq B \subseteq [m]$ and $t \in [T], k \in [K]$, we have $0 \leq v_{t,k}(A) \leq v_{t,k}(B)$.
We study the following three types of valuations:
\begin{itemize}
    \item Unit-demand: A function $v$ is unit-demand when $v(S) = \max_{i \in S} v(\{i\})$ for every $S \subseteq [m]$. When all functions $v_{t, k}$ in an online combinatorial allocation problem are unit-demand, the problem is known as the online stochastic matching.
    \item Submodular: A function $v$ is  submodular when $v(A \cup \{i\}) - v(A) \geq v(B \cup \{i\}) - v(B)$ for every $A \subseteq B \subseteq [m]$ and $i \in [m]$, or equivalently when $v(U) + v(V) \geq v(U \cap V) + v(U \cup V)$ for every $U, V \subseteq [m]$.
    \item XOS (fractionally subadditive): A function $v$ is XOS if there exists a set $\mathcal{A}$ of $m$ dimension vectors $a \in \R^m_{\geq 0}$, such that $v(S) = \max_{a \in \mathcal{A}} \sum_{i \in S} a_i$.
\end{itemize}

Note that the above classes of functions are listed in a hierarchical way, as XOS functions strictly contain submodular functions, which in turn strictly contain unit-demand functions.

\subsection{Online Configuration LP}
We fix the instance and consider the following online configuration LP. We refer to the third family of constraints as the online constraints. 
\begin{align}
    \text{maximize } \quad \quad  &    \quad\sum_{t \in [T]}  \sum_{k \in [K]} \sum_{S \subseteq [m]}~ v_{t,k}(S) \cdot x_{t,k, S}, \tag{$\lpon$} \label{program:ora-ub} \\
    \text{s.t.}\qquad \forall i \in [m], &  \quad   \sum_{t \in [T]} \sum_{k \in [K]} \sum_{S \ni i}~  x_{t,k,S}  ~\leq~ 1  \label{eq:exante-constraint}   \\
    \forall t \in [T], k \in [K], &   \quad \sum_{S \subseteq [m]}~ x_{t,k, S} ~\leq~ p_{t, k} \label{eq:probability-constraint} \\
    \forall t \in [T],  k \in [K], i \in [m],&  \quad  \sum_{S \ni i} x_{t, k, S} ~\leq~ p_{t, k} \cdot \left(1 - \sum_{t' < t} \sum_{k' \in [K]} \sum_{S \ni i} x_{t',k', S} \right) \label{eq:online-constraint} \\
    \forall t \in [T],  k \in [K], S \subseteq [m],&  \quad  0 ~\leq~ x_{t, k, S} ~\leq~ 1. \notag
\end{align}

We claim this LP provides an upper bound of the value of optimal online algorithm.
\begin{Observation}
\label{obs:lp-relaxation}
Given an instance of online combinatorial allocation problem, the optimal objective of $\lpon$ upper-bounds the reward achieved by the optimal online algorithm.
\end{Observation}
\begin{proof}
    Let $\tilde x_{t, k, S}$ be the probability that agent $t$ has a valuation function $v_{t, k}$, and the optimal online algorithm allocates subset $S$ to agent $t$. It is sufficient to show that $\{\tilde x_{t, k, S}\}$ is a feasible solution of $\lpon$. \eqref{eq:exante-constraint} is satisfied, because any online algorithm can allocate item $i$ at most once. \eqref{eq:probability-constraint} is also satisfied, by the definition of $\tilde x_{t,k,S}$. For constraint $\eqref{eq:online-constraint}$, note that $\sum_{S \in i} \tilde x_{t, k, S}$ represents the probability that agent $t$ has a valuation function $v_{t, k}$ and the optimal online algorithm allocates item $i$ to agent $t$. Since the realization of $v_t$ is independent of the event that item $i$ is already allocated before agent $t$ arrives, the LHS of \eqref{eq:online-constraint} should not be more than the probability that agent $t$ has a valuation $v_{t,k}$ times the probability that item $i$ is already allocated before agent $t$ arrives, which is exactly the RHS of \eqref{eq:online-constraint}.
\end{proof}

Since we aim to provide a polynomial-time algorithm for the online combinatorial allocation problem, we need to solve $\lpon$ efficiently. 
To achieve so, we assume access to the demand oracles of the valuation functions $v_{t,k}(\cdot)$.

\noindent{\bf Demand Oracle:} The demand oracle for function $v_{t,k}(S)$ takes a vector $(a_1, \cdots, a_m) \in \R^m_{\geq 0}$ as input and returns $\arg \max_S v_{t,k}(S) - \sum_{i \in S} a_i$.

That is, the demand oracle provides the utility-maximizing set of items under price $(a_1,\cdots,a_m)$. The following \Cref{lma:polytime-solvable} states that the online configuration LP can be solved efficiently, whose proof is deferred to \Cref{sec:prelim-appendix}.
\begin{restatable}{Lemma}{polysolve}
\label{lma:polytime-solvable}
    $\lpon$ is solvable in $\poly(T, m, K)$ time, assuming access to the \emph{demand oracles} of every $v_{t, k}(S)$. Furthermore, the total number of non-zero entries in the optimal solution is $\poly(T, m, K)$.
\end{restatable}

Let $\{x^*_{t, k, S}\}_{t \in [T], k \in [K], S \subseteq [m]}$ be the optimal solution. The above lemma states that $x^*$ has only $\poly(T, m, K)$ number of non-zero entries. 
Let 
\[
\lp(x^*) ~:=~ \sum_{t \in [T]}  \sum_{k \in [K]} \sum_{S \subseteq [m]}~ v_{t,k}(S) \cdot x^*_{t,k, S}
\]
be the objective value of solution $x^*$. \Cref{obs:lp-relaxation} implies that $\lp(x^*) \geq \opton$.

Without loss of generality, we assume the constraints \eqref{eq:exante-constraint} and \eqref{eq:probability-constraint} in $\lpon$ are both tight, i.e., we have $ \sum_{t \in [T]} \sum_{k \in [K]} \sum_{S \ni i}~  x^*_{t,k,S} = 1$ for every $i \in [m]$, and $\sum_{S \subseteq [m]} x^*_{t, k, S} = p_{t, k}$ for every $t \in [T], k \in [K]$. Assuming a tight \eqref{eq:exante-constraint} is without loss of generality, as we can add $m$ zero-value agents to the end of the instance, and use those newly added agents to cover the gap between $\sum_{t, k, S \ni i} x^*_{t, k, S}$ and $1$. For \eqref{eq:probability-constraint}, it is also without loss of generality, as we can arbitrarily increase the value of $x^*_{t, k, \varnothing}$ without changing the objective and the feasibility of constraint \eqref{eq:exante-constraint} and \eqref{eq:online-constraint}, but simultaneously making \eqref{eq:probability-constraint} tight. Furthermore, since we assume $\sum_{S \subseteq [m]} x^*_{t, k, S} = p_{t, k}$, for every $t \in [T], k \in [K]$, we may view $\{x^*_{t, k, S}/p_{t, k}\}$ as a distribution over subsets of $[m]$, i.e., the distribution returns $S \subseteq [m]$ with probability $x^*_{t, k, S}/p_{t, k}$.

\subsection{Single-Item Prophet Inequality}
The single-item online selection model shall appear frequently as a subroutine in our algorithm. This is a special case of the online combinatorial allocation problem with $m = 1$. And we refer to it as the single-item prophet inequality.

A single-item prophet inequality instance $\sipi$ contains $T$ distributions $\{\D^{(w)}_t(\sipi)\}_{t \in [T]}$. At time step $t \in [T]$, a non-negative value $w \sim \D^{(w)}_t(\sipi)$ is observed by the algorithm. The algorithm immediately decides whether to take the value and ends the game, or reject and continue to the next value. The goal is to maximize the value it takes. Without loss of generality, we assume the algorithm only accepts positive values. 
In this work, we only study instances $\sipi$ that satisfy the following ex-ante constraint:
\[
\sum_{t \in [T]} \mathop{\pr}\limits_{w \sim \D^{(w)}_t(\sipi)}\left[w > 0\right] ~\leq~ 1.
\]
Under this condition, the ex-ante benchmark of the instance, i.e., the LP value of the configuration LP without online constraint,  can be expressed as the sum of the expectations of each distribution. Formally, we define 
\[
\bm(\sipi) ~:=~ \sum_{t \in [T]} \mathop{\E}\limits_{w \sim \D^{(w)}_t(\sipi)}[w].
\]
Finally, we use $\opt(\sipi)$ to denote the expected value achieved by the optimal online algorithm for instance $\sipi$. Then, we have the following:
\begin{Lemma}
    \label{obs:sipi-half}
    For every single-item prophet inequality instance $\sipi$, it holds that
    \[
    \opt(\sipi) ~\geq~ 0.5 \cdot \bm(\sipi).
    \]
    Furthermore, the running time of the optimal online algorithm for $\sipi$ is polynomial in $T$ and the maximum support size of distributions $\{\D^{(w)}_t\}$.
\end{Lemma}
 The fact that the optimal online algorithm runs in poly-time follows from the observation that the optimal online algorithm is given by a threshold-based (reversed) dynamic program. To be more specific, the algorithm that achieves $\opt(\sipi)$ is parameterized by a group of thresholds $\tau^*_1, \cdots, \tau^*_T$, such that the algorithm only takes the $t$-th value when the value is at least $\tau^*_t$. After initiating $\tau^*_T = 0$, the remaining optimal thresholds are defined via the following backward dynamic programming:
 \[
 \tau^*_{t} ~:=~ \mathop{\E}\limits_{w \sim \D^{(w)}_{t+1}(\sipi)}\left[\max\{w, \tau^*_{t+1}\}\right].
 \]
 The fact that $\opt(\sipi) \geq 0.5 \cdot \bm(\sipi)$ follows from standard single-item prophet inequality, using ex-ante relaxation as the benchmark (e.g., OCRS \cite{FSZ-SODA16}). We omit the detailed proof of \Cref{obs:sipi-half}.

\subsection{Submodular Functions}
Our analysis rely on the following standard properties of submodular functions, whose proofs are deferred to \Cref{sec:prelim-appendix}.
\begin{restatable}{Claim}{restrict}
\label{clm:submod-restrict}
    Let $f(S): 2^{[m]} \to \R_{\geq 0}$ be a non-negative monotone submodular function. For any $A \subseteq [m]$, we define function $f(B\mid A) : 2^{[m]} \to \R_{\geq 0}$  as: 
    \[
    f(B\mid A) ~:=~ f(B \cup A) - f(A).
    \]
    Then $f(B \mid A)$ is also a non-negative monotone submodular function. Furthermore, for any $A \subseteq A' \subseteq [m]$, we have
    \[
    f(B \mid A) ~\geq~ f(B \mid A').
    \]
\end{restatable}
\begin{restatable}{Claim}{additive}
\label{clm:submod-additive}
    Let $f(S): 2^{[m]} \to \R_{\geq 0}$ be a non-negative monotone submodular function. For $A \subseteq [m]$, let $\D^{(A)}$ be a distribution over subsets of $A$. If $\pr_{S \sim D^{(A)}}[j \in S] \geq p$ holds for every $j \in A$, we have 
    \[
     \E_{S \sim D^{(A)}}[f(S)] ~\geq~ p \cdot f(A).
    \]
\end{restatable}

\begin{restatable}{Claim}{telescope}
\label{clm:submod-telescope}
Let $f(S): 2^{[m]} \to \R_{\geq 0}$ be a non-negative monotone submodular function. For any $A \subseteq B \subseteq [m]$, we have
\[
 f(A) ~\geq~ \sum_{i \in A} \big(f(B \cap [i]) - f(B \cap [i-1]) \big).
\]
\end{restatable}
We remark that 1) \Cref{clm:submod-additive} holds for the more general class of XOS functions verbatim, 2) \Cref{clm:submod-telescope} can be extended to XOS functions by substituting the corresponding value of the underlying additive vector for $f(B \cap [i]) - f(B \cap [i-1])$.
However, \Cref{clm:submod-restrict} does not hold for XOS functions. 
Though mathematically trivial, we make this remark so that it serves as a sanity check, giving evidence why our analysis is correct but only works for submodular agents.

\section{Submodular Philosopher Inequality}
\label{sec:submod}
In this section, we establish an $0.5 + \Omega(1)$ philosopher inequality for submodular agents. 
\begin{Theorem}
    \label{thm:lp-main}
    For online combinatorial allocation problem with submodular agents, there exists a constant $\epsilon \geq 10^{-16}$ and an  algorithm that achieves an expected reward of at least $(0.5 + \epsilon) \cdot \lp(x^*)$, where $x^*$ is the optimal solution of $\lpon$. Furthermore, the algorithm runs in $\poly(T, m, K)$ time.
\end{Theorem}

\subsection{Baseline Algorithm for Easy Instance Check}

We first introduce a baseline algorithm, which is guaranteed to obtain an expected reward of $0.5 \cdot \lp(x^*)$. This baseline algorithm is the algorithm we apply in the easy instance check, i.e., the first step of our final algorithm is to check whether running this baseline algorithm can achieve an expected reward of at least $(0.5 + \epsilon) \cdot \lp(x^*)$.

For every $t \in [T], k \in [K], S \subseteq [m]$, let
\begin{align*}
    w_{t, k, S}(i) ~:=~ v_{t, k}(S \cap [i]) - v_{t, k}(S \cap [i - 1])
\end{align*}
be the contribution of item $i$ to the value $v_{t, k}(S)$. Furthermore, let
\[
\lp_i(x^*) ~:=~ \sum_{t \in [T]} \sum_{k \in [K]} \sum_{S \subseteq [m]} x^*_{t, k, S} \cdot w_{t, K, S}(i)
\]
be the contribution of item $i$ to $\lp(x^*)$. 
By inspection, we have that
\[
\lp(x^*) = \sum_{i \in [m]} \lp_i(x^*) \quad \text{and} \quad v_{t,k}(S) = \sum_{i} w_{t,k,S}(i)~, \quad \forall t, k, S~.
\]
We remark that the second equality above is applying the idea that $v_{t,k}(S)$ is also an XOS function, and therefore we can decompose the value of $v_{t,k}(S)$ into the summation of $w_{t,k,S}(i)$ for $i \in S$. 

Our baseline algorithm is defined in \Cref{alg:submod-half}. 
The main idea of \Cref{alg:submod-half} is to reduce an online combinatorial allocation instance with submodular agents to multiple single-item prophet inequality instances. The algorithm is specifically designed for submodular agents, but the core property we use in \Cref{alg:submod-half} is that a submodular function is also an XOS function, and therefore the value can be decomposed into the sum of multiple additive values. We remark that this is a standard algorithm for combinatorial Prophet Inequalities, and a similar analysis can be found in literature, e.g., see \cite{CDHW-arxiv25}.

\begin{algorithm}[tbh]
\caption{\textsc{Baseline Algorithm that obtains $0.5 \lp(x^*)$}}
\label{alg:submod-half}
\begin{algorithmic}[1]
\State \textbf{input:} An online combinatorial allocation instance with submodular agents.
\State Solve $\lpon$ and get the optimal solution $\{x^*_{t, K, S}\}_{t \in [T], k \in [K], S \subseteq [m]}$.
\State Initiate $m$ single-item prophet inequality instances $\bssipi_1, \cdots, \bssipi_m$, such that for $i \in [m]$, $t \in [T]$, the distribution $\D^{(w)}_t(\bssipi_i)$ is defined as follows:
\[
\D^{(w)}_t(\bssipi_i)~:=~ \mathop{\pr}_{w \sim \D^{(w)}_t(\bssipi_i)}[w = w_{t,k, S}(i)] = x^*_{t,k,S}.
\]
\For{$t = 1, \ldots T$}
\State Algorithm reveals $v_t$. Let $k_t$ be the type of $v_t$, i.e., $v_t = v_{t, k_t}$.
\State Sample $S^{\req}_t $ from distribution $\{x^*_{t, k, S}/p_{t, k}\}$.
\State Initiate $S^\alg_t = \varnothing$.
\For{$i = 1 \cdots, m$}
\State Send value $w_{t, k, S^{\req}_t}(i)$ to instance $\bssipi_i$.
\State If $\bssipi_i$ accepts the value, add $i$ into $S^\alg_t$.
\EndFor
\State Allocate $S^\alg_t$ to agent $t$.
\EndFor
\end{algorithmic}
\end{algorithm} 

\begin{Lemma}
\label{lma:submod-half}
    The total expected reward achieved by \Cref{alg:submod-half} satisfies
    \[
    \E\left[\sum_{t \in [T]} v_t(S^\alg_t)\right] ~\geq~ 0.5 \cdot \lp(x^*).
    \]
    Furthermore, \Cref{alg:submod-half} runs in $\poly(T, m, K)$ time.
\end{Lemma}

To prove \Cref{lma:submod-half}, we need three ingredients. The first ingredient is provided in \Cref{obs:sipi-half}, which suggests that each single-item prophet inequality instance $\bssipi_i$ can achieve an expected reward of at least $0.5 \cdot \bm(\bssipi_i)$. The second ingredient is given by the following \Cref{clm:bmsipi}:

\begin{Claim}
    \label{clm:bmsipi} For every $i \in [m]$, we have $\bm(\bssipi_i) ~=~ \lp_i(x^*).$
\end{Claim}

\begin{proof}
    When agent $t$ arrives, instance $\bssipi_i$ independently receives value $w_{t, k, S}(i)$ with probability $p_{t, k} \cdot \frac{x^*_{t, k, S}}{p_{t,k}} = x^*_{t, k, S}$. Therefore, the expectation of the value sent to $\bssipi$ is $\sum_{k, S} x^*_{t, k, S} \cdot w_{t, k, S}(i)$, and we have
    \[
    \bm(\bssipi_i) ~=~ \sum_{t \in [T]} \sum_{k \in [K]} \sum_{S \subseteq[m]} x^*_{t, k, S} \cdot w_{t, k, S}(i) ~=~ \lp_i(x^*). \qedhere
    \]
\end{proof}

The last ingredient is the following \Cref{lma:reward-submod-to-si}, which connects the rewards from single-item prophet inequality instances to the performance of our algorithm:
\begin{Claim}
    \label{lma:reward-submod-to-si}
    In \Cref{alg:submod-half}, the expected reward $\E\left[\sum_{t \in [T]} v_t(S^{\alg}_t)\right]$ is at least $\sum_{i \in [m]} \opt(\bssipi_i)$.
\end{Claim}

\begin{proof}
    It is sufficient to show $v_t(S^\alg_t)$ is at least the total reward collected by all $\bssipi_i$ at step $t$. Then, summing over $t \in [T]$ and taking the expectations on both sides of the inequality proves \Cref{lma:reward-submod-to-si}.

    Consider the for loop in Line 8 of \Cref{alg:submod-half}. If $\bssipi_i$ accepts value $w_{t, k, S^\req_t}(i)$, the total reward of single-item instances increases by $w_{t, k, S^\req_t}(i)$; on the other hand, the increment of the submodular instance is 
    \[
    v_t(S^{\alg}_t \cap [i]) - v_t(S^{\alg}_t \cap [i - 1]) ~\geq~ v_t(S^{\req}_t \cap [i]) - v_t(S^{\req}_t \cap [i - 1])  ~=~ w_{t, k, S^\req_t}(i),
    \]
    where the inequality follows from submodularity and the fact that $S^{\alg}_t \subseteq S^{\req}_t$, because the single-item instance $\bssipi_i$ only accepts strictly positive value, which is possible only when $i \in S^{\req}_t$. Summing the increments over $i \in [m]$ finishes the proof.
\end{proof}

Now, we are ready to prove \Cref{lma:submod-half}:

\begin{proof}[Proof of \Cref{lma:submod-half}]
 We first verify that \Cref{alg:submod-half} is feasible, that is, each single item is allocated at most once.
Note that item $i$ is allocated to agent $t$ if and only if the corresponding $\bssipi_i$ takes the value sent to the instance. Since each $\bssipi_i$ takes at most one value, this naturally guarantees that each item $i$ is allocated at most once.

We also need to verify that the instances $\bssipi_1, \cdots, \bssipi_m$ are feasible, i.e., each $\bssipi_i$ should satisfy the ex-ante constraint. The ex-ante constraint holds because for every $i \in [m]$, we have
\begin{align*}
    \sum_{t \in [T]} \mathop{\pr}\limits_{v_t, S^\req_t}[w_{t, k, S^\req_t}(i) > 0] ~\leq~ \sum_{t \in [T]} \mathop{\pr}\limits_{v_t, S^\req_t}[i \in S^\req_t] ~=~ \sum_{t \in [T]} \sum_{k \in [K]} p_{t, k} \cdot \sum_{S \ni i} \frac{x^*_{t, k, S}}{p_{t, k}} ~\leq~ 1,
\end{align*}
where the last inequality follows from the constraint \eqref{eq:exante-constraint} in $\lpon$.

Next, we check the running time of \Cref{alg:submod-half}. Note that the set $\{w_{t, k, S}(i): x^*_{t, k, S} > 0\}$ includes all possible values sent to instance $\bssipi_i$. Therefore, the maximum support of the underlying distributions for $\bssipi_i$ is at most the number of non-zero entries in the optimal solution $\{x^*_{t, k, S}\}$. As \Cref{lma:polytime-solvable} states that $\{x^*_{t, k, S}\}$ has  $\poly(T,m,K)$ non-zero entries, \Cref{obs:sipi-half} ensures that each single-item instance $\bssipi_i$ runs in $\poly(T,m, K)$ time. For the remaining part of \Cref{alg:submod-half}, it can be directly checked that the running time is $\poly(T,m, K)$. Therefore, the total running time of \Cref{alg:submod-half} is $\poly(T,m, K)$.

Finally, for the expected reward achieved by \Cref{alg:submod-half}, combining \Cref{obs:sipi-half}, \Cref{clm:bmsipi}, and \Cref{lma:reward-submod-to-si} together gives
\[
\E\left[\sum_{t \in [T]} f_t(S^{\alg}_t)\right] ~\geq~ \sum_{i \in [m]} \opt(\bssipi_i) ~\geq~ 0.5 \cdot \sum_{i \in [m]} \bm(\bssipi_i) ~=~ 0.5 \lp(x^*),
\]
where the last inequality uses the fact that $\lp(x^*) = \sum_{i \in [m]} \lp_i(x^*)$.
\end{proof}

\subsection{Motivating Instance and Main Ideas}
\label{sec:motivate}
In this subsection, we provide an expanded and more detailed proof plan. To begin, note that if \Cref{alg:submod-half} achieves an expected reward of $(0.5 + \epsilon) \cdot \lp(x^*)$, i.e., the easy instance check passes,  then we are satisfied with its performance. Thus, we focus on instances where \Cref{alg:submod-half} achieves an expected reward of at most $(0.5 + \epsilon) \cdot \lp(x^*)$. What additional structure of the instance can we exploit under the assumption that ``\Cref{alg:submod-half} achieves at most $(0.5 + \epsilon) \cdot \lp(x^*)$''?

Let us take a closer look at \Cref{alg:submod-half}. Its performance guarantee essentially stems from the fact that $\opt(\bssipi_i) \geq 0.5 \lp_i(x^*)$ for each single-item instance $\bssipi_i$. If every instance $\bssipi_i$ satisfies $\opt(\bssipi_i) \geq (0.5 + \epsilon) \cdot \lp_i(x^*)$, then by \Cref{lma:reward-submod-to-si}, \Cref{alg:submod-half} achieves an expected reward of at least $(0.5 + \epsilon) \cdot \sum_{i \in [m]} \lp_i(x^*) = (0.5 + \epsilon) \cdot \lp(x^*)$.

Conversely, if \Cref{alg:submod-half} achieves an expected reward of at most $(0.5 + \epsilon) \cdot \lp(x^*)$, then almost every single-item instance $\bssipi_i$ must be ``hard''—that is, $\opt(\bssipi_i)$ cannot be much greater than $0.5 \lp_i(x^*)$. To build intuition around this case, we present \Cref{exp:hard}, which exactly satisfies the condition discussed above. We remark that an instance similar to \Cref{exp:hard} is also provided in \cite{BDL-EC22}, which is used for proving the tightness of their analysis.

\begin{Example}
\label{exp:hard}
    Let $\delta \to 0$ be an arbitrarily small constant, such that $\delta^{-1}$ is an integer. Consider the following online combinatorial allocation instance with unit-demand (and therefore submodular) agents. Let $m = \delta^{-1}$, $T = m + 1$. For $t = 1, 2, \cdots, m$, agent $t$ arrives with probability $p_t = 1 - \delta$ with unit-demand valuation $v_t(S) = \one[t \in S] \cdot 1$. For $t = T$, agent $T$ arrives with probability $p_T = 1$ with unit-demand valuation $v_t(S) = \one[S \neq \varnothing] \cdot \delta^{-1}$.
\end{Example}
For \Cref{exp:hard}, it is easy to verify that the optimal solution to $\lpon$ is as follows: we have $x^*_{t, \{t\}} = 1 - \delta$ for each $t \in [m]$, and $x^*_{T, \{i\}} = \delta$ for each $i \in [m]$\footnote{We omit the dimension $k$ in the optimal solution, as \Cref{exp:hard} includes only Bernoulli agents.}. Now consider running \Cref{alg:submod-half} on \Cref{exp:hard}. For each single-item instance $\bssipi_i$, the values it receives match those in the standard hard instance for the single-item prophet inequality: a first value of $1$ arrives with probability $1 - \delta$ (corresponding to agent $t = i$), and a second value of $\delta^{-1}$ arrives with probability $\delta$ (corresponding to agent $T$).

This setup gives rise to the following intuition: whenever $\opt(\bssipi_i)$ is close to $0.5 \cdot \bm(\bssipi_i)$, the instance should resemble the standard worst-case instance of the single-item prophet inequality. To formalize this idea, we identify two structural properties that characterize when $\bssipi_i$ is hard:

\noindent \textbf{Property (\romannumeral1).} The value-probability pairs ${(w_{t, k, S}(i), x^*_{t, k, S})}$ are assumed to follow a “free-deterministic” structure. That is, we can decompose the set $\{(w_{t, k, S}(i), x^*_{t, k, S})\}$ into two components such that each contributes approximately $0.5 \lp_i(x^*)$ to the benchmark when summing $\sum w_{t, k, S}(i) \cdot x^*_{t, k, S}$ over that part.

The first component, referred to as the free part, consists of pairs where the total probability mass $\sum x^*_{t, k, S}$ is close to $0$. Each pair $(w_{t, k, S}(i), x^*_{t, k, S})$ in this group has a very small probability $x^*_{t, k, S}$ but a high value $w_{t, k, S}(i)$, and can be accepted essentially “for free” whenever it appears. This part corresponds to the $(\delta^{-1}, \delta)$ value-probability pair in each single-item instance of \Cref{exp:hard}.

The second component, referred to as the deterministic part, comprises pairs where the total probability mass $\sum x^*_{t, k, S}$ is close to $1$. In this group, the values $w_{t, k, S}(i)$ are all close to $0.5 \lp_i(x^*)$, providing a stable, predictable contribution to the benchmark. This part corresponds to the $(1, 1 - \delta)$ pair in each single-item instance of \Cref{exp:hard}.

\noindent \textbf{Property (\romannumeral2).} Note that the free-deterministic structure described in Property (\romannumeral1) does not account for the arrival order of values. Property (\romannumeral2) addresses this by asserting that, when $\bssipi_i$ is hard, most of the deterministic part must arrive before the free part. If this condition is not met, we can show that \Cref{alg:submod-half} already performs well—achieving an expected reward of at least $(0.5 + \epsilon) \cdot \lp_i(x^*)$. This is because, for each $\bssipi_i$, it suffices to select all of the free part and the portion of the deterministic part that arrives after the free part.

We formalize the above properties through quantified lemmas and provide detailed proofs in \Cref{sec:frdt}. What remains is to design a polynomial-time algorithm for the case when both properties hold. To build intuition, we first give a brief overview of our algorithm, using \Cref{exp:hard} as a guiding example.

We begin with an algorithm that only allocates the free part of the instance. When applied to \Cref{exp:hard}, since the optimal solution of the online LP has $x^*_{T, \{i\}} = \delta$ for each $i$, the algorithm simply skips the first $T - 1$ agents, and allocates an item to the last agent uniformly at random. It can be verified that this algorithm achieves $0.5 \cdot \lp(x^*)$.

Next, an idea that seems worth considering for general instances is to sample the free part twice and allocate the union of the two free parts to an agent. Intuitively, this idea looks appealing because the free part offers extremely high value with extremely small probability. Assuming the free part of each item has probability at most $\delta$, this “double sampling” algorithm achieves a total reward of at least $1 - 2\delta$ times the expected value provided by the union of the two free parts.

Ideally, we would hope that the union of two free parts yields an expected reward strictly better than that achieved by sampling once. However, it can be immediately verified that this idea breaks down even for \Cref{exp:hard}: applying the double sampling strategy here simply means the algorithm skips the first $T - 1$ agents, samples two items uniformly and independently, and allocates their union to the last agent. Since the last agent is unit-demand, it makes no sense to allocate two items to the last agent and this algorithm  achieves at most $0.5 \lp(x^*)$.

To refine the above algorithm, we consider applying the double sampling idea to only half of the items, leading to the core ``half-double sampling'' idea of our algorithm. Specifically, let $ Q \subseteq [m] $ be a random subset of items, where each item $ i $ is included in $ Q $ independently with probability 0.5. We apply the double sampling rule to the items in $ Q $; that is, an item $ i \in Q $ can be allocated to an agent if it appears in the union of two free parts. For the remaining items in $ [m] \setminus Q $, we construct a new online combinatorial allocation instance and allocate these items using the baseline algorithm (\Cref{alg:submod-half}). The final allocation to an agent is given by the union of the items allocated from $ Q $ and from $ [m] \setminus Q $.

Let us evaluate the performance of this new algorithm on \Cref{exp:hard}. We first sample a subset $Q$ of size approximately $\frac{1}{2\delta}$. For each item $i \in [m] \setminus Q$, following the baseline algorithm, it is simply allocated to agent $t = i$ and yields value $1$. The items in $Q$ are only considered when the last agent arrives: we sample two items uniformly and independently, and allocate to the last agent if either of the sampled items falls in $Q$. 

In terms of performance, the items in $[m] \setminus Q$ contribute $1/2$ of the total value from the deterministic part, while the items in $Q$ contribute $3/4$ of the total value from the free part, since the last agent is matched to at least one item in $Q$ with probability exactly $3/4$. Therefore, compared to the $0.5 \lp(x^*)$ benchmark, our algorithm gains an additional $1/4$ of the total value from the free part, leading to a $0.625$-approximation.

The core of our proof is to show that the above algorithm remains effective when applied to a general online combinatorial allocation instance with submodular agents. To be specific, we show that allocating the union of two free parts for items in $Q$ still captures $3/4$ of the total value from the free part, resulting in an additional gain of $1/4 \cdot 0.5 \lp(x^*) = 1/8 \cdot \lp(x^*)$. The proof of this argument critically relies on the online constraint \eqref{eq:online-constraint} in $\lpon$ to show that each item in $Q$ has nearly zero probability of falling into a sampled free part. As a result, the two sampled free parts are almost disjoint, yielding a structure similar to the unit-demand case. On the other hand, this extra gain is halved in the general case, where the $\frac{1}{2}$ factor comes from  constructing a new online combinatorial allocation instance for the items in $[m] \setminus Q$. Therefore, the algorithm achieves an additional $1/16 \cdot \lp(x^*)$ gain over the $0.5 \lp(x^*)$ benchmark, leading to a $0.5625$-approximation.

We present the details of the above argument in \Cref{sec:wesmall} and our core half-double sampling algorithm in \Cref{alg:wesmall}. We note that the idea breaks down when agents have XOS valuation functions. We shall construct a modified version of \Cref{exp:hard}, in which the last agent's unit-demand function is replaced by an XOS function, yielding a hard instance with a $0.5$ integrality gap for the online configuration LP. We refer the reader to \Cref{sec:xos} for further details.

\subsection{Free-Deterministic Decomposition and Proof of \Cref{thm:lp-main}}
\label{sec:frdt}

In this section, we quantify the two properties we discussed in  \Cref{sec:motivate}. We first give the following \Cref{lma:property1}, which quantifies Property (\romannumeral1):

\begin{restatable}{Lemma}{property}
\label{lma:property1}
    Let $\{x^*_{t, k, S}\}$ be the optimal solution of $\lpon$ for a submodular prophet inequality instance. Let $\bssipi_i$ be the single-item prophet inequality instance created in \Cref{alg:submod-half}. If we have $\sum_{i \in [m]} \opt(\bssipi_i) \leq (0.5 + \epsilon) \cdot \lp(x^*)$ for some $\epsilon < 10^{-4}$, there exists a subset of items $U \subseteq [m]$, such that $\sum_{i \in [m] \setminus U} \lp_i(x^*) \leq 4\epsilon^{1/4} \cdot \lp(x^*)$. Furthermore, define
    \begin{align*}
            &~\fr = \{(t, k, S, i): x^*_{t, k, S} > 0 \land i \in S \land i \in U \land  w_{t, k, S}(i) >  (0.5 + 2\epsilon^{1/4}) \cdot \lp_i(x^*)\} \\
            \quad \text{and} \quad &~ \dt = \{(t, k, S, i): x^*_{t, k, S} > 0 \land i \in S \land (t, k, S, i) \notin \fr\}.
    \end{align*}
    Then, 
    \begin{itemize}
        \item For every $i \in [m]$, we have 
        \[
        \sum_{t, k, S} x^*_{t, k, S} \cdot \one[(t, k, S, i) \in \fr] ~\leq~ \epsilon^{1/4}.
        \]
        \item For every $i \in U$, we have 
        \[
        (0.5 - 2\epsilon^{1/4}) \cdot \lp_i(x^*) \leq \sum_{t, k, S} x^*_{t, k, S} \cdot w_{t, k, S}(i) \cdot \one[(t, k, S, i) \in \fr] \leq (0.5 + 2\epsilon^{1/4}) \cdot \lp_i(x^*).
        \]
    \end{itemize}
\end{restatable}

In \Cref{lma:property1}, the subset $U \subseteq [m]$ refers to the set of items for which $\opt(\bssipi_i)$ is not significantly larger than $0.5 \bm_i$, allowing the application of a free-deterministic decomposition (formally defined in \Cref{lma:free_prob}). When condition $\sum_{i \in [m]} \opt(\bssipi_i) \leq (0.5 + \epsilon) \cdot \lp(x^*)$ holds, we can show that $\sum_{i \in U} \lp_i(x^*)$ is approximately equal to $\lp(x^*)$. 

The subsets $\fr$ (short for ``free'') and $\dt$ (short for ``deterministic'') represent the collections of free and deterministic items, respectively. Note that $\fr$ and $\dt$ contain tuples of the form $(t, k, S, i)$; that is, whether an item $i$ is classified as free or deterministic can only be determined once the identity of the agent, the type of the agent, and even the specific subset that we plan to allocate to this agent are all fixed.

The proof of \Cref{lma:property1} relies on the following \Cref{lma:free_prob}:

\begin{restatable}{Lemma}{freeprob}
\label{lma:free_prob}
Given a single-item prophet inequality instance $\sipi$ with $\{\D^{(w)}_t\}_{t \in [T]}$ being the underlying value distributions.  For every $\mu<0.5$ and $ \beta>0.5+\mu$, if no online algorithm can achieve an expected reward of more than $(0.5 + \mu) \cdot \bm(\sipi)$ from instance $\sipi$, we have
\begin{flalign*}
1) & \sum_{t \in [T]}  ~\mathop{\pr}\limits_{w \sim \D^{(w)}_t}\left[w > \beta \cdot \bm(\sipi)\right] ~\leq~  \delta~; & \\
2) & \sum_{t \in [T]}  ~\mathop{\E}\limits_{w \sim \D^{(w)}_t}\left[\one[w > \beta \cdot \bm(\sipi)] \cdot w \right] ~\leq~   \frac{0.5+\mu}{1-\delta} \cdot \bm(\sipi),~&
\end{flalign*}
where $\delta =  \sqrt{\frac{4\mu}{\beta - 0.5 - \mu}}$.
\end{restatable}

\Cref{lma:free_prob} is the main free-deterministic decomposition we wish to apply. It follows from Lemma 4.1 in \cite{stoc/STW25} with two differences: \begin{itemize}
    \item The original lemma in \cite{stoc/STW25} assumes the instance only contains Bernoulli values, while \Cref{lma:free_prob} also considers non-Bernoulli requests.
    \item The original lemma only requires the instance does not admit good \emph{threshold-based} algorithms that achieves $(0.5 + \mu) \cdot \bm(\sipi)$, while \Cref{lma:free_prob} requires $\opt(\sipi) \leq (0.5 + \mu) \cdot \bm(\sipi)$.
\end{itemize} 
As the proof of the Bernoulli case is given in \cite{stoc/STW25}, we prove \Cref{lma:free_prob} via providing a reduction from non-Bernoulli case to Bernoulli case. We defer the proofs of \Cref{lma:free_prob} and \Cref{lma:property1} to \Cref{sec:submod-appendix}.

Next, we formalize Property~(\romannumeral2) by quantifying it. Our main idea is to further decompose the free value set $\fr$ into an early part (denoted $\freasy$) and a late part (denoted $\frhard$). Let $\epseh$ be a small constant to be optimized later. We define: \begin{align} \freasy := \left\{(t, k, S, i) \in \fr : \sum_{s < t} \sum_{k' \in [K]} \sum_{S' \ni i} x^*_{s, k', S'} \leq 1 - \epseh \right\} \quad \text{and} \quad \frhard := \fr \setminus \freasy. \label{eq:def-freh} \end{align}

Intuitively, $\freasy$ consists of free values that arrive \textbf{early}: The cumulative probability mass that arrives before values in $\freasy$ is at most $1 - \epseh$. In contrast, $\frhard$ captures the remaining, \textbf{late}-arriving free values. The parameter $\epseh$ thus serves as a threshold distinguishing “early” from “late”.

Note that when Property~(\romannumeral2) holds, most deterministic values arrive before free values. As a result, the set $\freasy$ should be nearly empty. To formalize this intuition, we define: 
\begin{align} \we := \sum_{(t, k, S, i)} x^*_{t, k, S} \cdot w_{t, k, S}(i) \cdot \one[(t, k, S, i) \in \freasy] \label{eq:def-we} 
\end{align} 
as the total weight of values in $\freasy$. The following lemma shows that $\we$ must be small; otherwise, \Cref{alg:submod-half} would already achieve a high reward.

\begin{restatable}{Lemma}{welarge}
\label{lma:welarge}
    Assume the optimal solution $x^*$ of $\lpon$  satisfies \Cref{lma:property1}. Then, \Cref{alg:submod-half} achieves an expected reward of at least 
    \[
    (0.5 - 13\epsilon^{1/4} - \epseh^2) \cdot \lp(x^*) + \epseh \cdot \we.
    \]
\end{restatable}

The only remaining case is when both Property~(\romannumeral1) and Property~(\romannumeral2) hold. In this case, the following lemma shows that there exists an algorithm that achieves an expected reward 
of $0.5625 \cdot \lp(x^*)$ as $\epsilon$, $\epseh$, and $\we$ approach zero, matching the approximation ratio discussed in \Cref{sec:motivate}.

\begin{Lemma}
    \label{lma:wesmall}
    Assume the optimal solution $x^*$ of $\lpon$  satisfies \Cref{lma:property1}. There exists an algorithm that achieves an expected reward of 
    \[
    (0.5625 - \epseh/8 - 6.5\epsilon^{1/4}) \cdot \lp(x^*) - 0.625 \cdot \we.
    \] 
    Furthermore, the algorithm runs in $\poly(T, m, K)$ time.
\end{Lemma}

We defer the proofs of \Cref{lma:wesmall} and \Cref{lma:welarge} to \Cref{sec:wesmall} and \Cref{sec:welarge} respectively, and first prove \Cref{thm:lp-main}:

\begin{proof}
    Consider to take
    \[
    \epsilon ~=~ 10^{-16} \quad \text{and} \quad \epseh ~=~ 0.033.
    \]
    Our final algorithm runs via the following steps:
    \begin{enumerate}[Step 1:]
        \item  Solve $\lpon$ in $\poly(T,m,K)$ time, and get the optimal solution $x^*$.
        \item Perform the easy instance check via pretending to run \Cref{alg:submod-half}. Let $\bssipi_i$ be the revised single-item prophet inequality instance given by \Cref{alg:submod-half}. Check if $\sum_{i \in [m]} \opt(\bssipi_i) \geq (0.5 + \epsilon) \cdot \lp(x^*)$. If so, run \Cref{alg:submod-half}.
        \item If the easy instance check in Step 2 fails, run \Cref{alg:submod-half} with probability $\frac{0.625}{0.625 + \epseh}$, and run the algorithm provided in \Cref{lma:wesmall}  (see \Cref{alg:wesmall} for the corresponding algorithm) with probability $\frac{\epseh}{0.625 + \epseh}$.
    \end{enumerate}

    Each step in the above process can be directly verified to run in $\poly(T, m, K)$ time. Therefore, the overall running time is also $\poly(T, m, K)$. It remains to show that the above algorithm achieves an expected reward of $(0.5 + \epsilon) \cdot \lp(x^*)$. For Step 2, \Cref{lma:reward-submod-to-si} guarantees that \Cref{alg:submod-half} achieves an expected reward of at least $(0.5 + \epsilon) \cdot \lp(x^*)$ when $\sum_{i \in [m]} \opt(\bssipi_i) \geq (0.5 + \epsilon) \cdot \lp(x^*)$ holds. For Step 3, by \Cref{lma:welarge} and \Cref{lma:wesmall}, the expected performance of Step 3 is
    \begin{align*}
        &\left((0.5 - 13\epsilon^{1/4} - \epseh^2) \cdot \lp(x^*) + \epseh \cdot \we \right) \cdot \frac{0.625}{0.625 + \epseh} \\
        &~+ \left( (0.5625 - \epseh/8 - 6.5\epsilon^{1/4}) \cdot \lp(x^*) - 0.625 \cdot \we\right) \cdot \frac{\epseh}{0.625 + \epseh} \\
        ~\geq~& 0.5\lp(x^*) + \frac{\lp(x^*)}{0.625 + \epseh} \cdot \left(0.0625 \epseh - 0.75\epseh^2 - 8.125 \epsilon^{1/4} - 6.5 \epseh \cdot \epsilon^{1/4}\right) \\
        ~\geq~& (0.5 + \epsilon) \cdot \lp(x^*),
    \end{align*}
    where the last inequality can be verified numerically.
\end{proof}

\section{Algorithm for Small $\we$: Proof of \Cref{lma:wesmall}}
\label{sec:wesmall}

In this section, we prove \Cref{lma:wesmall}. We show that  \Cref{alg:wesmall} is the half-double sampling algorithm for \Cref{lma:wesmall}.

\begin{algorithm}[tbh]
\caption{\textsc{Algorithm for small $\we$}}
\label{alg:wesmall}
\begin{algorithmic}[1]
\State \textbf{input:} An online combinatorial allocation instance with submodular agents.
\State Let $\{x^*_{t, k, S}\}$ be the optimal solution of $\lpon$ that satisfies \Cref{lma:property1}.
\State Let $Q \subseteq [m]$ be a random set, such that each $i \in [m]$ is added into $Q$ w.p. 0.5 independently.
\State Initiate $m$ single-item prophet inequality instances $\wesipi_1, \cdots, \wesipi_m$.
\State Let $R_1 = [m]$ be the set of remaining items for agent $1$.
\For{$t = 1, \ldots T$}
\State Algorithm reveals $v_t$. Let $k_t$ be the type of $v_t$, i.e., $v_t = v_{t, k_t}$.
\State Sample $A_t, B_t$ both from distribution $\{x^*_{t, k_t, S}/p_{t, k_t}\}$ independently.
\State Define $A^F_t = \{i \in A_t: (t, k_t, A_t, i) \in \frhard\}$ and $B^F_t = \{i \in B_t: (t, k_t, B_t, i) \in \frhard\}$.
\State Let $S^F_t = Q \cap (A^F_t \cup B^F_t)$ be the union of ``free'' parts for items in $Q$.
\State Let $S^\req_t = A_t \setminus Q$ be the subset we request for items in $[m] \setminus Q$.
\State Let $S^\get_t$ be the subset that agent $t$ gets from PI instances. We initiate $S^\get_t = \varnothing$.
\State Define function $f_t(S) = v_t(S \mid S^F_t)$. 
\For{$i = 1 \cdots, m$}
\State Send value $f_t(S^\req_t \cap [i]) - f_t(S^\req_t \cap [i-1])$ to instance $\wesipi_i$.
\State If $\wesipi_i$ accepts the value, add $i$ into $S^\get_t$.
\EndFor
\State Allocate $S_t = (S^F_t \cap R_t) \cup S^\get_t$ to agent $t$, and update $R_{t+1} = R_t \setminus S_t$.
\EndFor
\end{algorithmic}
\end{algorithm}

The definition of $\D^{(w)}_t(\wesipi_i)$ for the single-item prophet inequality instance $\wesipi_i$ is omitted in \Cref{alg:wesmall} due to its complexity. Here, we provide the full definition to complete \Cref{alg:wesmall}:
\begin{itemize}
    \item For $i \in Q$, the distribution $\D^{(w)}_t(\wesipi_i)$ always generates value $0$ with probability $1$.
    \item For $i \in [m] \setminus Q$,  with probability $ x^*_{t, k, A} \cdot x^*_{t, k, B} / p_{t,k}$ (for each $A, B \subseteq [m]$), the distribution generates the value:
    \[
    v_t\Big( \big( A \setminus Q\big) \cap [i] \mid Q \cap (A^F \cup B^F)\Big) - v_t\Big( \big( A  \setminus Q\big) \cap [i - 1] \mid Q \cap (A^F \cup B^F)\Big),
    \]
    where $A^F = \{i \in A : (t, k, A, i) \in \frhard\}$ and $B^F = \{i \in B : (t, k, B, i) \in \frhard\}$.
\end{itemize} 
Since the definition is involved, we recommend that readers refer directly to \Cref{alg:wesmall} for a concrete understanding of the values passed to $\wesipi_i$. Our analysis will also rely directly on \Cref{alg:wesmall}, rather than on the explicit definition of $\D^{(w)}_t(\wesipi_i)$ given above.

\medskip
Now, we provide the proof of \Cref{lma:wesmall}.
\paragraph{Feasibility of \Cref{alg:wesmall}.} To prove \Cref{lma:wesmall}, we start from verifying that \Cref{alg:wesmall} is feasible - each single item $i$ is allocated at most once, and each single-item prophet inequality instance $\wesipi_i$ is feasible.

We first show that every item is allocated at most once. For the items in $Q$, it can be allocated only when $i \in S^F_t$ and $i \in R_t$ simultaneously hold. Since $R_t$ is the set of available items for agent $t$, $i$ should be allocated at most once.  For the items in $[m] \setminus Q$, it is allocated only when $i \in S^\get_t$, which is possible only when the corresponding instance $\wesipi_i$ takes the value. Since each $\wesipi_i$ can take at most one value, each item $i \in [m] \setminus Q$ can be allocated at most once.

Next, we verify that each single-item prophet inequality instance $\wesipi_i$ is feasible, i.e., it satisfies the ex-ante constraint. Note that the value sent to $\wesipi_i$ is possible to be non-zero only when $i \in S^\req_t$, which is true when  $i \in A_t$ happens. Since we have
    \begin{align*}
        \pr[i \in A_t] ~&=~ \sum_{k \in [K]} p_{t, k} \cdot \sum_{S \ni i} \frac{x^*_{t, k, S}}{p_{t, k}} ~=~ \sum_{k \in [K]} \sum_{S \ni i} x^*_{t, k, S}, 
    \end{align*}
     there must be
    \begin{align*}
        \sum_{t \in [T]} \pr[i \in S^\req_t] ~&\leq~ \sum_{t \in [T]}  \sum_{k \in [K]} \sum_{S \ni i} x^*_{t, k, S} ~\leq~ 1,
    \end{align*}
    where the last inequality follows from the constraint \eqref{eq:exante-constraint} in $\lpon$. Therefore, each single-item prophet inequality instance $\wesipi_i$ satisfies the ex-ante constraint.

\paragraph{Running time analysis.} Next, we show that \Cref{alg:wesmall} runs in polynomial time. We begin by proving that each single-item instance $\wesipi_i$ can be executed in $\poly(T, m, K)$ time by bounding the total number of possible values passed to $\wesipi_i$. Note that the function $f_t$ depends on $v_t$, $A_t$, and $B_t$. Since the number of possible tuples $(v_t, A_t, B_t)$ is at most the cube of the number of non-zero entries in the solution ${x^*_{t, k, S}}$, and \Cref{lma:polytime-solvable} states that this number is $\poly(T, m, K)$, the total number of distinct $f_t$ is also bounded by $\poly(T, m, K)$. Therefore, the support size of the underlying distribution for each instance $\wesipi_i$ is at most $\poly(T, m, K)$, and each $\wesipi_i$ can be executed in polynomial time.

On the other hand, it is straightforward to verify that the remaining steps in \Cref{alg:wesmall} can also be completed in $\poly(T, m, K)$ time. Therefore, the total running time of \Cref{alg:wesmall} is $\poly(T, m, K)$.

\paragraph{Decomposing the total reward.} Finally, we analyze the total reward collected by \Cref{alg:wesmall}.   For agent $t$, after fixing  the randomness of $Q, v_t, A_t, B_t$, we have
\begin{align*}
    v_t(S_t) ~&=~ v_t(S^F_t \cap R_t) + v_t(S^\get_t \cup( S^F_t \cap R_t)) -v_t(S^F_t \cap R_t) \\
    ~&=~ v_t(S^F_t \cap R_t) + v_t(S^\get_t \mid S^F_t \cap R_t)  \\
    ~&\geq~ v_t(S^F_t \cap R_t) + v_t(S^\get_t \mid S^F_t )  ~=~ v_t(S^F_t \cap R_t) + f_t(S^\get_t)
\end{align*}
where the inequality follows from the fact that $S^F_t \cap R_t \subseteq S^F_t$ and \Cref{clm:submod-additive}. Taking the expectation over the randomness of $Q, v_t, A_t, B_t$, we lower-bound the reward from \Cref{alg:wesmall} as
\begin{align}
    \E\left[\sum_{t \in [T]} v_t(S^F_t \cap R_t)\right] ~+~  \E\left[\sum_{t \in [T]} f_t(S^\get_t)\right]. \label{eq:wesmall-reward1}
\end{align}

We further lower bound the two terms in \eqref{eq:wesmall-reward1} separately. The first term can be bounded via the following \Cref{clm:fr-probsmall}:

\begin{Claim}
    \label{clm:fr-probsmall}
    For every $i \in Q$ and $t \in [T]$, we have $\pr[i \in R_t] \geq 1 - 2\epsilon^{1/4}$.
\end{Claim}

\begin{proof}
    Each item $i \in Q$ is possible to be allocated only when event $i \in A^F_t$ or $i \in B^F_t$ happens. 
    For every $i \in Q$ and $t \in [T]$, we have
    \[
    \pr[i \in A^F_t] ~=~ \sum_{k \in [K]} p_{t, k} \cdot \sum_{S \subseteq[m]} \frac{x^*_{t, k, S}}{p_{t,k}} \cdot \one[(t, k, S, i) \in \frhard] 
    ~\leq~ \sum_{k \in [K]} \sum_{S \subseteq[m]} x^*_{t, k, S} \cdot \one[(t, k, S, i) \in \fr].
    \]
    Now taking the union bound over the events $i \in A^F_t$ and $i \in B^F_t$ for all $t \in [T]$ (note that $A^F_t$ and $B^F_t$ are identically distributed), we have
    \begin{align*}
        \pr[i \text{ is allocated}] ~\leq~ 2 \cdot \sum_{t \in [T]} \sum_{k \in [K]} \sum_{S \subseteq[m]} x^*_{t, k, S} \cdot \one[(t, k, S, i) \in \fr] ~\leq~ 2\epsilon^{1/4},
    \end{align*}
    where the last inequality follows from \Cref{lma:property1}. Therefore, \Cref{clm:fr-probsmall} holds from the observation that $\pr[i \in R_t] \geq 1 - \pr[i \text{ is allocated}]$.
\end{proof}

Now, we apply \Cref{clm:fr-probsmall} to lower-bound $\E[v_t(S^F_t \cap R_t)]$. Since $S^F_t \cap R_t$ is a distribution that outputs a subset of $S^F_t$, such that for each $i \in S^F_t$, we have $\pr[i \in (S^F_t \cap R_t)] \geq 1 - 2\epsilon^{1/4}$. Then, \Cref{clm:submod-additive} guarantees
\begin{align}
    \E\left[\sum_{t \in [T]} v_t(S^F_t \cap R_t)\right] ~\geq~ \sum_{t \in [T]} (1 - 2\epsilon^{1/4}) \cdot \E[v_t(S^F_t)]. \label{eq:wesmall-reward2}
\end{align}

For the second term in \eqref{eq:wesmall-reward1}, we lower-bound it via the following \Cref{clm:calling-subroutine}:

\begin{Claim}
    \label{clm:calling-subroutine}
    We have
    \begin{align}
        \E\left[\sum_{t \in [T]} f_t(S^\get_t)\right] ~\geq~ \frac{1}{2} \cdot \E\left[\sum_{t \in [T]} f_t(S^\req_t)\right] ~=~ \frac{1}{2} \sum_{t \in [T]} \E\left[v_t(S^\req_t \cup S^F_t) - v_t(S^F_t)\right] \label{eq:wesmall-reward3}
    \end{align}
\end{Claim}

\begin{proof}
    We first show that 
    \begin{align}
        \E\left[\sum_{t \in [T]} f_t(S^\get_t)\right] ~\geq~ \sum_{i \in [m]} \opt(\wesipi_i). \label{eq:calling-subroutine1}
    \end{align}
    To prove \eqref{eq:calling-subroutine1}, note that an item $i$ is added into $S^\get_t$ if and only if the corresponding prophet inequality instance $\wesipi_i$ accepts value  $f_t(S^\req_t \cap [i]) - f_t(S^\req_t \cap [i-1])$. Consider to compare the increment of both sides in \eqref{eq:calling-subroutine1}.
     When the for loop in Line 14 of \Cref{alg:wesmall} reaches $i$ and instance $\wesipi_i$ accepts value $f_t(S^\req_t \cap [i]) - f_t(S^\req_t \cap [i-1])$, the RHS of \eqref{eq:calling-subroutine1} increases by exactly $f_t(S^\req_t \cap [i]) - f_t(S^\req_t \cap [i-1])$, and the LHS of \eqref{eq:calling-subroutine1} increases by
    \[
    f_t(S^\get_t \cap [i]) - f_t(S^\get_t \cap [i-1]) ~\geq~ f_t(S^\req_t \cap [i]) - f_t(S^\req_t \cap [i-1]),
    \]
    where the inequality follows from the fact that $f_t$ is submodular (which is guaranteed by \Cref{clm:submod-restrict}),  and the fact that $S^\get_t \subseteq S^\req_t$, which is guaranteed by the fact that instance $\wesipi_i$ only accepts non-zero values, and therefore $i \in S^\get_t$ implies $i \in S^\req_t$. Therefore, \eqref{eq:calling-subroutine1} holds. This further gives 
    \begin{align*}
        \E\left[\sum_{t \in [T]} f_t(S^\get_t)\right] ~\geq~ \sum_{i \in [m]} \opt(\wesipi_i) ~&\geq~ \frac{1}{2} \cdot \sum_{i \in [m]} \bm(\wesipi_i) \\
        ~&\geq~ \frac{1}{2} \cdot \E\left[\sum_{t \in [T]} \sum_{i \in [m]} \left(f_t(S^\req_t \cap [i]) - f_t(S^\req_t \cap [i-1]) \right)\right] \\
        ~&=~ \frac{1}{2} \cdot \E\left[\sum_{t \in [T]} f_t(S^\req_t)\right],
    \end{align*}
    where the second inequality uses \Cref{obs:sipi-half}.
\end{proof}

Applying \eqref{eq:wesmall-reward2} and \eqref{eq:wesmall-reward3} to \eqref{eq:wesmall-reward1}, we further lower-bound the performance of \Cref{alg:wesmall} as
\begin{align}
    (0.5 - 2\epsilon^{1/4}) \cdot \sum_{t \in [T]}  \E[v_t(S^F_t)] +  0.5  \cdot \sum_{t \in [T]} \E[v_t(S^\req_t \cup S^F_t)] \label{eq:wesmall-reward-final}
\end{align}

We remark that \eqref{eq:wesmall-reward-final} explicitly reflects the idea discussed in \Cref{sec:motivate}: compared to the benchmark $0.5 \lp(x^*)$, the first term in \eqref{eq:wesmall-reward-final} represents the source of the additional gain. Here, the coefficient $\frac{1}{2}$ arises from the observation that the items in $Q$ contribute a factor $1$ (corresponding to the first term in \eqref{eq:wesmall-reward1} and \Cref{clm:fr-probsmall}), while the extra gain is halved when constructing the new online combinatorial allocation instance (corresponding to the second term in \eqref{eq:wesmall-reward1} and \Cref{clm:calling-subroutine}).

Next, we lower-bound the two terms in \eqref{eq:wesmall-reward-final} separately.

\paragraph{Bounding $\E[v_t(S^F_t)]$.} The first term in \eqref{eq:wesmall-reward-final} is where we need to show that applying double sampling to half of the items gains at least $3/4$ times the total free value. Assuming $\epsilon$, $\epseh$, and $\we$ approach zero, our goal is to establish that 
\[
\E[v_t(S^F_t)] = \E[v_t(Q \cap (A^F_t \cup B^F_t))] \geq \frac{3}{4} \cdot \E[v_t(A^F_t)],
\]
as we will further show that $v_t(A^F_t)$ collects almost all the free values. 

We first show that a standard approach for lower-bounding $\E[v_t(S^F_t)]$ does not give the desired inequality. Recall from the discussion in \Cref{sec:motivate} that the double sampling idea does not work on all items; that is, when $Q$ is not considered,
\[
 \E[v_t(A^F_t \cup B^F_t)] ~\geq~  \E[v_t(A^F_t)]
\]
is the best lower bound we can obtain for $\E[v_t(A^F_t \cup B^F_t)]$.  If we naively apply \Cref{clm:submod-additive} to $\E[v_t(S^F_t)]$, we obtain the following.
\[
\E[v_t(S^F_t)] ~=~ \E[v_t(Q \cap (A^F_t \cup B^F_t))] ~\geq~ \frac{1}{2} \cdot  \E[v_t(A^F_t \cup B^F_t)] ~\geq~ \frac{1}{2} \cdot \E[v_t(A^F_t)],
\]
which falls short by a factor of $1/4$ relative to our target. Therefore, a more careful analysis is required to show that $\E[v_t(S^F_t)] \geq \frac{3}{4} \cdot \E[v_t(A^F_t)]$.

We now present the proof of this bound. Fix $t$, we have
\begin{align}
    \mathop{\E}\limits_{Q, v_t, A_t, B_t}[v_t(S^F_t)] ~&=~ \mathop{\E}\limits_{Q, v_t, A_t, B_t} [v_t(Q \cap (A^F_t \cup B^F_t))] \notag\\
    ~&=~ \mathop{\E}\limits_{Q, v_t, A_t, B_t} [v_t\big(~ (Q \cap A^F_t) \cup ((B^F_t \setminus A^F_t) \cap Q)~\big) - v_t(Q \cap A^F_t) + v_t(Q\cap A^F_t)] \label{eq:first-1}
\end{align}

We first bound the expectation of the first two terms in \eqref{eq:first-1}. 
Note that the randomness from $Q$ is independent to the realization of $v_t, A_t, B_t$. Consider to first realize $v_t, A_t, B_t$, and then determine whether $i \in Q$ for items in $A^F_t$ (by flipping independent fair coins), and leave the randomness of whether $i \in Q$ for items in $B^F_t \setminus A^F_t$ undecided.  \Cref{clm:submod-restrict} guarantees that the function $v_t(S \mid Q \cap A^F_t)$ is submodular, and the expectation of the first two terms in \eqref{eq:first-1} can be written as
\begin{align*}
    \mathop{\E}\limits_{v_t, A_t, B_t, Q \cap A^F_t} \left[\mathop{\E}\limits_{Q \cap (B^F_t \setminus A^F_t)} [v_t(Q \cap (B^F_t \setminus A^F_t) \mid Q \cap A^F_t)]\right] &\geq \mathop{\E}\limits_{v_t, A_t, B_t, Q \cap A^F_t} \left[0.5 v_t(B^F_t \setminus A^F_t\mid Q \cap A^F_t) \right],
\end{align*}
where the inequality follows from the fact that subset $Q \cap (B^F_t \setminus A^F_t)$ can be viewed as adding each item in $B^F_t \setminus A^F_t$ into $Q$ w.p. $0.5$ independently, the submodularity of function $v_t(S \mid Q \cap A^F_t)$, and \Cref{clm:submod-additive}. Plugging the above inequality into \eqref{eq:first-1}, we have
\begin{align}
    \mathop{\E}\limits_{Q, v_t, A_t, B_t}[v_t(S^F_t)] ~&\geq~ 0.5 \mathop{\E}\limits_{Q, v_t, A_t, B_t}[v_t((Q \cap A^F_t) \cup (B^F_t \setminus A^F_t)) + v_t(Q \cap A^F_t)] \notag \\
    ~&\geq~ 0.5 \mathop{\E}\limits_{v_t, A_t, B_t}[v_t(B^F_t \setminus A^F_t)] +  0.5 \mathop{\E}\limits_{Q, v_t, A_t}[v_t(Q \cap A^F_t)]  \label{eq:first-2}
\end{align}
For the second term in \eqref{eq:first-2}, note that $v_t$ is a submodular function, and subset $Q \cup A^F_t$ can be viewed as adding each item in $A^F_t$ into $Q$ w.p. $0.5$ independently. Then, \Cref{clm:submod-additive} guarantees
\begin{align}
    0.5 \mathop{\E}\limits_{Q, v_t, A_t}[v_t(Q \cap A^F_t)] ~\geq~ 0.25 \cdot \mathop{\E}\limits_{v_t, A_t}[v_t(A^F_t)]. \label{eq:first-3}
\end{align}
For the first term in \eqref{eq:first-2}, we bound it via the following \Cref{clm:aft-small}:

\begin{Claim}
\label{clm:aft-small}
    For every $t \in [T], k \in [K]$, let $A$ be a random set independently drawn from distribution $\{x^*_{t, k, S}/p_{t, k}\}$, and let $A^F = \{i \in A: (t, k, A, i) \in \frhard\}$. Then, for every $i \in [m]$, we have $\pr[i\in A^F]\leq\epseh$.
\end{Claim}

\begin{proof}
    Fix $t, k, i$, we have
    \begin{align*}
        \pr[i \in A^F] ~&=~ \sum_{S \ni i} \frac{x^*_{t, k, S}}{p_{t, k}} \cdot \one\left[1 -\sum_{s < t} \sum_{k' \in [K]} \sum_{S' \ni i} x^*_{s, k', S'} < \epseh \land (t, k, S, i) \in \fr \right] \\
        ~&\leq~ \one\left[1 -\sum_{s < t} \sum_{k' \in [K]} \sum_{S' \ni i} x^*_{s, k', S'} < \epseh\right] \cdot \frac{\sum_{S \ni i} x^*_{t, k, S}}{p_{t, k}} \\
        ~&\leq~ \one\left[1 -\sum_{s < t} \sum_{k' \in [K]} \sum_{S' \ni i} x^*_{s, k', S'} < \epseh\right] \cdot \left(1 -\sum_{s < t} \sum_{k' \in [K]} \sum_{S' \ni i} x^*_{s, k', S'}\right) \\
        ~&\leq~ \epseh,
    \end{align*}
    where the third line follows from the online constraint \eqref{eq:online-constraint} in $\lpon$.
\end{proof}

Note that after fixing the randomness of $v_t$ and $B_t$, subset $B^F_t \setminus A^F_t$ can be viewed as a distribution that outputs a subset of $B^F_t$. Since \Cref{clm:aft-small} guarantees that $i \in A^F_t$ with probability at most $\epseh$, every $i \in B^F_t$ should still belong to $B^F_t \setminus A^F_t$ w.p. $1 - \epseh$. Then, the submodularity of function $v_t$ and \Cref{clm:submod-additive} guarantees
\begin{align}
    0.5 \mathop{\E}\limits_{v_t, A_t, B_t}[v_t(B^F_t \setminus A^F_t)] ~\geq~ 0.5 \cdot (1 - \epseh) \cdot \mathop{\E}\limits_{v_t, B_t}[v_t(B^F_t)] ~=~ 0.5 \cdot (1 - \epseh) \cdot \mathop{\E}\limits_{v_t, A_t}[v_t(A^F_t)], \label{eq:first-4}
\end{align}
where the last equality follows from the fact that $A^F_t$ and $B^F_t$ are identically distributed. Applying \eqref{eq:first-3} and \eqref{eq:first-4} to \eqref{eq:first-2} gives
\begin{align}
     \E[v_t(S^F_t)] ~\geq~ (0.75 - 0.5 \epseh) \cdot \E[v_t(A^F_t)]. \label{eq:first-final}
\end{align}

\paragraph{Bounding $\E[v_t(S^\req_t \cup S^F_t)]$.} Recall that $S^\req_t = A_t  \setminus Q$, and $S^F_t = (A^F_t \cup B^F_t) \cap Q \supseteq A^F_t \cap Q$. Then, since each agent has a monotone valuation function,  we have
\begin{align}
    \mathop{\E}\left[v_t(S^\req_t \cup S^F_t)\right] ~&\geq~ \mathop{\E}\limits_{Q, v_t, A_t}\left[v_t\big(~(A_t \setminus Q) \cup (A^F_t \cap Q) ~\big)\right] \notag \\
    ~&=~ \mathop{\E}\limits_{Q, v_t, A_t}\left[v_t\big(~((A_t \setminus A^F_t) \setminus Q) \cup A^F_t ~\big) - v_t(A^F_t) + v_t(A^F_t)\right]. \label{eq:second-1}
\end{align}
We first bound the first two terms in \eqref{eq:second-1}. Fix the randomness of $v_t, A_t$. \Cref{clm:submod-restrict} guarantees that function $v_t(S \mid A^F_t) $ is submodular, and the expectation of first two terms in \eqref{eq:second-1} can be written as
\[
\mathop{\E}\limits_{v_t, A_t}\left[\mathop{\E}\limits_{Q}\left[v_t\big(~(A_t \setminus A^F_t) \setminus Q \mid A^F_t~\big)\right]\right] ~\geq~ \mathop{\E}\limits_{v_t, A_t}\left[0.5 \cdot v_t\big(~(A_t \setminus A^F_t) \mid A^F_t ~\big)\right],
\]
where the inequality follows from the fact that subset $(A_t \setminus A^F_t) \setminus Q$ can be viewed as adding each item in $A_t \setminus A^F_t$ into $[m] \setminus Q$ w.p. $0.5$ independently, the submodularity of function $v_t(S \mid A^F_t) $, and \Cref{clm:submod-additive}. Plugging the above inequality into \eqref{eq:second-1}, we have
\begin{align}
    \mathop{\E}\left[v_t(S^\req_t \cup S^F_t)\right] ~&\geq~  \mathop{\E}\limits_{v_t, A_t}\left[0.5 \cdot v_t\big(~(A_t \setminus A^F_t) \cup A^F_t~\big) - 0.5 \cdot v_t(A^F_t) + v_t(A^F_t)\right] \notag \\
    ~&=~ 0.5 \cdot \mathop{\E}\limits_{v_t, A_t}\left[v_t(A_t) \right] + 0.5 \cdot \mathop{\E}\limits_{v_t, A_t}\left[ v_t(A^F_t)\right]  \label{eq:second-final}.
\end{align}

\paragraph{Putting everything together.} By applying \eqref{eq:first-final} and \eqref{eq:second-final} to \eqref{eq:wesmall-reward-final}, we can lower-bound the performance of \Cref{alg:wesmall} as
\begin{align}
    (0.5 - 2\epsilon^{1/4}) \cdot \sum_{t \in [T]} \mathop{\E}\limits_{v_t, A_t} [(1.25 - 0.5 \epseh) \cdot v_t(A^F_t) + 0.5v_t(A_t)] \label{eq:putting-final}
\end{align}
It remains to bound $\E[v_t(A_t)]$ and $\E[v_t(A^F_t)]$. For $\E[v_t(A_t)]$, we have
\begin{align*}
    \sum_{t \in [T]} \mathop{\E}\limits_{v_t, A_t} [v_t(A_t)] ~=~ \sum_{t \in [T]} \sum_{k \in [K]} p_{t, k} \sum_{S \subseteq [m]} \frac{x^*_{t, k, S}}{p_{t, k}} \cdot v_{t, k}(S) ~=~ \lp(x^*).
\end{align*}
For $\E[v_t(A^F_t)]$, we have
\begin{align*}
    \sum_{t \in [T]} \mathop{\E}\limits_{v_t, A_t} [v_t(A^F_t)] ~&=~ \sum_{t \in [T]} \sum_{k \in [K]} p_{t, k} \sum_{A_t \subseteq [m]} \frac{x^*_{t, k, A_t}}{p_{t, k}} \cdot v_{t, k}(A^F_t) \\
    ~&\geq~ \sum_{t \in [T]} \sum_{k \in [K]}  \sum_{A_t \subseteq [m]} x^*_{t, k, A_t} \cdot \sum_{i:(t, k, A_t, i) \in \frhard} w_{t, k, A_t}(i) \\
    ~&=~ \sum_{t, k, S} x^*_{t, k, S} \cdot \left(\sum_{i: (t, k, S, i) \in \fr} w_{t, k, S}(i) - \sum_{i: (t, k, S, i) \in \freasy} w_{t, k, S}(i)\right) \\
    ~&\geq~ \sum_{i \in U} (0.5 - 2\epsilon^{1/4})\lp_i(x^*) - \we \\
    ~&\geq~ (0.5 - 4\epsilon^{1/4}) \cdot \lp(x^*) - \we,
\end{align*}
where the second line uses the fact that $A^F_t \subseteq A_t$ and \Cref{clm:submod-telescope}, the third line uses the fact that $\frhard = \fr \setminus \freasy$,  the forth line uses \Cref{lma:property1} together with the definition of $\we$, and the last line uses the fact that $\sum_{i \notin U} \lp_i(x^*) \leq 4\epsilon^{1/4} \cdot \lp(x^*)$. Plugging the bounds on $\E[v_t(A_t)]$ and $\E[v_t(A^F_t)]$ into \eqref{eq:putting-final}, we finally lower-bound the performance of \Cref{alg:wesmall} as
\begin{align*}
    &~(0.5 - 2\epsilon^{1/4}) \cdot \left((1.25 - 0.5\epseh) \cdot ((0.5 - 4\epsilon^{1/4}) \lp(x^*) - \we) + 0.5 \lp(x^*)\right) \\
    ~\geq&~ (0.5625 - \epseh/8 - 6.5\epsilon^{1/4}) \cdot \lp(x^*) - 0.625 \cdot \we.
\end{align*}

\section{Hard Instances for Combinatorial Philosopher Inequalities}
\label{sec:xos}

In this section, we show the integrality gaps of $\lpon$ for online combinatorial allocation problem when all agents fall in the class of XOS. To be specific, we prove the following:

\begin{Theorem}
    \label{thm:xos-hard} 
    For online combinatorial allocation problem with XOS agents, and any constant $\gamma > 0$, there exist a hard instance such that no online algorithm can achieve an expected reward of $(0.5 + \gamma) \cdot \lp(x^*)$.
\end{Theorem}

To prove \Cref{thm:xos-hard}, we introduce the following \Cref{exp:xos-hard}:

\begin{Example}
\label{exp:xos-hard}
    Let $\delta \to 0$ be a sufficiently small constant, such that $\delta^{-1}$ is an integer.  Consider the following online combinatorial allocation problem with XOS agents: Let $m = \delta^{-3}$, $T = m + 1$. For $t = 1, 2, \cdots, m$, agent $t$ arrives with probability $p_t = 1-\delta$ with unit-demand valuation $v_t(S) = \one[t \in S] \cdot 1$. For $t = T$, agent $T$ arrives with probability $1$ and a \emph{random} valuation function $v_T(S)$ defined as follows: 
    \begin{itemize}
        \item Let $U_1, \cdots, U_{\delta^{-1}}$ be a partition of item set $[m]$ taken uniformly at random, such that $|U_j| = \delta^{-2}$ for every $j \in [\delta^{-1}]$.
        \item Define $v_T(S) = \frac{1}{\delta} \cdot \max_{j \in [\delta^{-1}]} \left|U_j \cap S \right|$.
    \end{itemize}
\end{Example}

\begin{proof}[Proof of \Cref{thm:xos-hard}]
    We prove \Cref{thm:xos-hard} by showing \Cref{exp:xos-hard} is the desired hard instance.
    \paragraph{Feasibility of \Cref{exp:xos-hard}.} We first verify that \Cref{exp:xos-hard} is a feasible instance with XOS agents. For $t \leq m$, function $v_t(S)$ is unit-demand, and therefore must be XOS. For the last $v_T(S)$, it is taking maximum over a group of additive vectors $\boldsymbol{a}_1, \cdots, \boldsymbol{a}_{\delta^{-1}}$, such that for $j \in [\delta^{-1}]$, we have $\boldsymbol{a}_{j} = (\one[1 \in U_j], \one[2 \in U_j], \cdots, \one[m \in U_j])$. Therefore, the last valuation function $v_T(S)$ is also XOS.  

    \paragraph{A feasible solution $\tilde x$. } Next, we give a feasible solution $\tilde x$ of $\lpon$ for \Cref{exp:xos-hard}. We will further show that no online algorithm can achieve an expected reward greater than $(0.5 + \gamma) \cdot \lp(\tilde x)$. This is already sufficient to prove \Cref{thm:xos-hard}, as the value of $\lp(x^*)$ is at least $\lp(\tilde x)$. 

    For $t = 1, 2, \cdots, m$, we set $\tilde x_{t, \{t\}} = 1 - \delta$.\footnote{We omit the dimension $K$ for agents $t \in [m]$, as these are Bernoulli agents with a single type.} For the last agent $T$, let 
\[
N ~:=~ \frac{1}{(\delta^{-1})!} \cdot  \binom{m}{\underbrace{\delta^{-2}, \delta^{-2}, \ldots, \delta^{-2}}_{\delta^{-1} \text{ groups}}} 
\]
be the number of different partitions $\{U_1, \cdots, U_{\delta^{-1}}\}$  that partition $[m]$ into $\delta^{-1}$ subsets, with each subset containing $r$ items. Then, the last agent has exactly $N$ different types, and each type occurs with probability $1/N$. For each $k \in [N]$, define $\{U^k_1, \cdots, U^k_{\delta^{-1}}\}$ to be the corresponding partition of type $k$. Then, we set $\tilde x_{T, k, U^k_j} = \delta/N$ for every $k \in [N], j \in [\delta^{-1}]$.

 Now, we show that $\tilde x$ is feasible. We first verify constraint \eqref{eq:exante-constraint}. For every $i \in [m]$, we have
\begin{align*}
    \sum_{t \in [m]} \sum_{S \ni i} \tilde x_{t, S} + \sum_{k \in [N]} \sum_{S \ni i} \tilde x_{T, k, S} ~=~ 1 - \delta + N \cdot \frac{\delta}{N} ~=~ 1, 
\end{align*}
where we use the fact that $\tilde x_{i, \{i\}} = 1- \delta$ is the only non-zero term for $t = 1, \cdots, m$, and for $t = T$, since $\{U^k_1, \cdots, U^k_{\delta^{-1}}\}$ is a partition of $[m]$, for every type $k \in [N]$ there exists exactly one $S = U^k_j$ that satisfies $i \in S$ and $\tilde x_{T, k, S} \neq 0$. Since the corresponding $\tilde x_{T, k, S}$ equals to $\delta / N$, the second summation contributes $N \cdot \delta/N = \delta$. Therefore, constraint \eqref{eq:exante-constraint} is satisfied.

For constraint \eqref{eq:probability-constraint}, it's easy to check that our construction guarantees \eqref{eq:probability-constraint} becomes equal for every agent and every type. 

For constraint \eqref{eq:online-constraint}, it easy to check that the constraint holds for $t \in [m]$. For the last agent $t = T$, note that for every $i \in [m]$, we have
\[
1 - \sum_{t' < T} \sum_{S \ni i} \tilde x_{t', S} ~=~ \delta.
\]
Therefore, \eqref{eq:online-constraint} holds because we have
\begin{align*}
    \sum_{S \ni i} \tilde x_{t, k, S} ~=~ \frac{\delta}{N} ~=~ p_{t, k} \cdot \delta
\end{align*}
for every $i \in [m]$ and $k \in [N]$, where the first equality follows from the fact that $\{U^k_1, \cdots, U^k_{\delta^{-1}}\}$ is a partition of $[m]$, and there exists exactly one $S = U^k_j$ that satisfies $i \in S$ and $\tilde x_{T, k, S} \neq 0$.

\paragraph{Proving \Cref{thm:xos-hard}.} Finally, we prove \Cref{thm:xos-hard} by showing no online algorithm can achieve an expected reward of more than $(0.5 + \gamma) \cdot \lp(\tilde x)$, where we have
\begin{align*}
    \lp(\tilde x) ~=~ \sum_{t \in [m]} \tilde x_{t, \{t\}} \cdot 1 + \sum_{k \in [N]} \sum_{j \in [\delta^{-1}]} \tilde x_{T, k, U^k_j} \cdot \delta^{-3}  ~=~ m \cdot (2 - \delta).
\end{align*}
Here, the first equality uses the fact that $v_{t,k}(U^k_j) = \delta^{-1} \cdot |U^k_j| = \delta^{-3}$, and the second equality uses the fact that $\delta^{-3} = m$.

Now, we upper-bound the performance of the optimal online algorithm. Let $R$ be a (possibly random) subset of items that is still not allocated when agent $T$ arrives. Then, the optimal strategy for any online algorithm is to simply allocate $R$ to agent $T$, and collect $v_T(R)$.  Let $|R| = r$. We show that when $\delta$ is sufficiently small and $n$ is sufficiently large, there must be
\begin{align}
    \mathop{\E}\limits_{v_T}\left[v_T(R)\right] ~\leq~ r + \frac{2}{\delta^{2.75}}. \label{eq:vtsmall}
\end{align}

The proof of \eqref{eq:vtsmall} relies the following concentration inequality for sampling without replacement process:
\begin{Lemma}[Corollary 2.4.2 in \cite{GM-MOR16} and Theorem 2.14.19 in \cite{weakConvergence}]
\label{Bernstein}
    Let $Y = \{y_1, \cdots, y_u\}$ be a set of real numbers in the interval $[0, M]$. Let $S$ be a random subset of $Y$ of size $v$ and let $Y_S = \sum_{i \in S} y_i$. Setting $\mu = \frac{1}{u} \sum_{i} y_i$, 
    %and $\sigma^2 = \frac{1}{u} \sum_i (y_i - \mu)^2$, 
    we have that for every $\tau > 0$,
    \[
    \pr\left[\left|Y_s - v \cdot \mu \right| \geq \tau \right] ~\leq~ 2\exp \left(-\frac{\tau^2}{M(4v\mu + \tau)} \right).
    \]
\end{Lemma}

We use \Cref{Bernstein} to prove \eqref{eq:vtsmall}: Fix subset $R$ with $|R| = r$. Let $\{U_1, \cdots, U_{\delta^{-1}}\}$ be a partition of $[m]$ generated uniformly at random, such that $|U_j| = \delta^{-2}$ for every $j \in [\delta^{-1}]$. Note that for a single subset $U_j$, it is generated by a random process that samples $\delta^{-2}$ items from the set $[m]$ without replacement. Therefore, applying \Cref{Bernstein} with $u = m$, $M = 1$, $y_i = \one[i \in R]$ for $i \in [m]$, and $v = \delta^{-2}$, we have
\begin{align*}
    \pr\left[|R\cap U_j| - \delta^{-2} \cdot  \frac{r}{m} \geq \delta^{-1.75} \right] ~&\leq~ 2 \exp \left(-\frac{\delta^{-3.5}}{4\delta^{-2} \cdot r/m + \delta^{-1.75}} \right) \\
    ~&\leq~ 2 \exp \left(-\frac{1}{5\delta^{1.5}} \right) \\
    ~&\leq~ 2 \exp \left(-10\log(\delta^{-1}) \right) ~\leq~ 2\delta^{10},
\end{align*}
where the second line uses the fact that $r = |R| \leq m$, and the last line holds when $\delta$ is sufficiently small. With the above inequality, we can lower-bound the expectation of $v_T(R)$ as the following: We have
\begin{align*}
    \delta \cdot \mathop{\E}\limits_{v_T}\left[v_T(R)\right] ~&=~   \mathop{\E}\limits_{\{U_1, \cdots, U_{\delta^{-1}}\}}\left[\max_{j \in [\delta^{-1}]} |U_j \cap R|\right] \\
    ~&\leq~ \delta^{-2} \cdot r/m + \delta^{-1.75} + \delta^{-2} \cdot \pr\left[\max_{j \in [\delta^{-1}]} |U_j \cap R| > \delta^{-2} \cdot r/m + \delta^{-1.75}\right] \\
    ~&\leq~ \delta^{-2} \cdot r/m + \delta^{-1.75} + \delta^{-2} \cdot \frac{1}{\delta} \cdot 2\delta^{10} \\
    ~&\leq~  \delta^{-2} \cdot r/m + 2\delta^{-1.75},
\end{align*}
where we use the fact that $\max_{j \in [\delta^{-1}]} |U_j \cap R| \leq \max_{j \in [\delta^{-1}]} U_j = \delta^{-2}$ in the second line, and the union bound over all $j \in [\delta^{-1}]$ in the third line. Rearranging the above inequality and applying $m = \delta^{-3}$ proves \eqref{eq:vtsmall}.

With \eqref{eq:vtsmall}, we are ready to finish the proof of \Cref{thm:xos-hard}. For any online algorithm for \Cref{exp:xos-hard}, let $R$ be the random set that the algorithm leaves for agent $T$, and let $r = |R|$. Then, the total reward gained by the algorithm is at most
\begin{align*}
    \E_R \left[(m - r) + r + 2\delta^{-2.75} \right] ~=~ m \cdot \left(1 + 2\delta^{0.25}\right) ~\leq~ (2 - \delta) \cdot (0.5 + \gamma) \cdot m~=~ (0.5 +\gamma) \cdot\lp(\tilde x),
\end{align*}
where the inequality holds when we take a sufficiently small $\delta$ that satisfies $0.5\delta + 2\delta^{0.25} \leq \gamma$.
\end{proof}

\bigskip

\bibliographystyle{alpha}
\bibliography{ref.bib,bib.bib}

\newpage

\appendix

\section{Algorithm for Large $\we$: Proof of \Cref{lma:welarge}}
\label{sec:welarge}

In this section, we prove \Cref{lma:welarge}, which is restated below for convenience:

\welarge*

\begin{proof}
Let $\bssipi_i$ be the single-item instance in \Cref{alg:submod-half} for item $i$. Recall that \Cref{lma:reward-submod-to-si} suggested that the total reward achieved by \Cref{alg:submod-half} is at least $\sum_{i \in [m]} \opt(\bssipi_i)$. Therefore, it's sufficient to show
\begin{align*}
    \sum_{i \in [m]} \opt(\bssipi_i) ~\geq~ (0.5 - 13\epsilon^{1/4} - \epseh^2) \cdot \lp(x^*) + \epseh \cdot \we.
\end{align*}

Define notation
\[
\we(i) ~:=~ \sum_{(t, k, S)} x^*_{t, k, S} \cdot w_{t, k, S}(i) \cdot \one[(t, k, S, i) \in \freasy]
\]
to be the amount that item $i$ contributes to $\we$. Then, there must be $\sum_{i \in [m]} \we(i) = \we$. For $i \in [m]$ such that $\we(i) \leq 0.5\epseh \cdot \lp_i(x^*)$, we simply lower-bound $\opt(\bssipi_i)$ as $0.5 \lp_i(x^*)$. For the remaining items that satisfy $\we(i) > 0.5\epseh \cdot \lp_i(x^*)$, we will give an explicit algorithm and lower-bound the value of $\opt(\bssipi_i)$ as the performance of that explicit algorithm. Note that condition $\we(i) > 0.5\epseh \cdot \lp_i(x^*)$ can be satisfied only when $i \in U$, because there must be $\we(i) = 0$ for $i \notin U$. Therefore, we assume $i \in U$ in the remaining part of the proof.

To give the explicit algorithm for $\bssipi_i$, we start from defining $t_i \in [T]$ to be the index of an agent that satisfies the following: 
    \[
    \sum_{t = t_i}^T \sum_{k \in [K]} \sum_{S \ni i}x^*_{t, k, S} \geq \epseh \quad \text{and} \quad \sum_{t = t_i + 1}^T \sum_{k \in [K]} \sum_{S \ni i}x^*_{t, k, S}< \epseh.
    \]
Such index $t_i$ must exist, as when $t_i = T$, the second condition must be satisfied (as the summation is $0$); and when $t_i = 1$, the first condition must satisfied (as we recall the summation is assumed to be $1$), and the summations in both conditions are monotone. With parameter $t_i$, the algorithm for instance $\bssipi_i$ is as follows:
\begin{itemize}
    \item When $t < t_i$ and value $w_{t, k, S}(i)$ arrives, pick the value if $i$ is available and $(t, k, S, i) \in \fr$.
    \item When $t \geq t_i$ and value $w_{t, k, S}(i)$ arrives, pick the value if $i$ is available.
\end{itemize}
Since $\opt(\bssipi_i)$ is lower-bounded by the performance of the above algorithm, we have
\begin{align}
    \opt(\bssipi_i) ~\geq~& \sum_{t = 1}^{t_i - 1} \sum_{k \in [K]} \sum_{S \ni i} x^*_{t, k, S} \cdot w_{t, k, S}(i) \cdot \one[(t, k, S, i) \in \fr] \cdot \one[i \text{ is available for }t] \notag\\
    &+ \sum_{t = t_i}^{T} \sum_{k \in [K]} \sum_{S \ni i} x^*_{t, k, S} \cdot w_{t, k, S}(i)  \cdot \pr[i \text{ is available for }t] \notag \\
    ~=~& \sum_{t = 1}^{t_i} \sum_{k \in [K]} \sum_{S \ni i} x^*_{t, k, S} \cdot w_{t, k, S}(i) \cdot \one[(t, k, S, i) \in \fr] \cdot \pr[i \text{ is available for }t] \label{eq:welarge-parta} \\
    &+  \sum_{k \in [K]} \sum_{S \ni i} x^*_{t_i, k, S} \cdot w_{t_i, k, S}(i)  \cdot \one[(t_i, k, S, i) \in \dt] \cdot \one[i \text{ is available for }t_i] \label{eq:welarge-partb} \\
    &+ \sum_{t = t_i + 1}^T \sum_{k \in [K]} \sum_{S \ni i} x^*_{t, k, S} \cdot w_{t, k, S}(i)  \cdot \one[(t, k, S, i) \in \fr] \cdot \pr[i \text{ is available for }t]\label{eq:welarge-partc} \\
    &+ \sum_{t = t_i + 1}^T \sum_{k \in [K]} \sum_{S \ni i} x^*_{t, k, S} \cdot w_{t, k, S}(i)  \cdot \one[(t, k, S, i) \in \dt] \cdot \pr[i \text{ is available for }t]. \label{eq:welarge-partd} 
\end{align}
Intuitively speaking, we decompose the total reward collected by our algorithm into four parts:
\begin{itemize}
    \item \eqref{eq:welarge-parta} represents the free value we collect from agent $1$ to agent $t_i$. 
    \item \eqref{eq:welarge-partb} represents the deterministic value we collect from agent $t_i$.
    \item \eqref{eq:welarge-partc} represents the free value we collect from agent $t_i + 1$ to agent $T$.
    \item \eqref{eq:welarge-partd} represents the deterministic value we collect from agent $t_i + 1$ to agent $T$.
\end{itemize}

Next, we lower-bound \eqref{eq:welarge-parta}, \eqref{eq:welarge-partb}, \eqref{eq:welarge-partc}, and \eqref{eq:welarge-partd} separately.

\paragraph{Bounding \eqref{eq:welarge-parta}.} We first lower-bound $\pr[i \text{ is available for }t]$. Note that when $t < t_i$, the algorithm only picks free values. Therefore, for $t \in [1, t_i]$, the probability that $i$ is still available for $t$ is at least the probability that no free value arrives between $[1, t_i - 1]$. By the union bound, this is at least
\[
1 - \sum_{t = 1}^T \sum_{k \in [K]} \sum_{S \ni i} x^*_{t, k, S} \cdot \one[(t, k, S, i) \in \fr] ~\geq~ 1 - \epsilon^{1/4},
\]
where the last inequality follows from \Cref{lma:property1}. Then, we have
\[
\eqref{eq:welarge-parta} ~\geq~ (1 - \epsilon^{1/4}) \cdot \sum_{t = 1}^{t_i} \sum_{k \in [K]} \sum_{S \ni i} x^*_{t, k, S} \cdot w_{t, k, S}(i) \cdot \one[(t, k, S, i) \in \fr].
\]
Next, note that for $t \leq t_i$, there must be
\[
\sum_{s < t} \sum_{k' \in [K]} \sum_{S' \ni i} x^*_{s, k', S'} ~\leq~ 1 - \sum_{t' =  t_i}^T \sum_{k' \in [K]} \sum_{S' \ni i} x^*_{t', k', S'} ~\leq~ 1 - \epseh,
\]
where the first inequality uses the fact that we assume $\sum_{t', k', S' \ni i} x^*_{t', k', S'} = 1$. Symmetrically, for $t > t_i$, there must be
\[
\sum_{s < t} \sum_{k' \in [K]} \sum_{S' \ni i} x^*_{s, k', S'} ~\geq~ 1 - \sum_{t' =  t_i + 1}^T \sum_{k' \in [K]} \sum_{S' \ni i} x^*_{t', k', S'} > 1 - \epseh.
\]
Then, the condition $(t, k, S, i) \in \freasy$ is equivalent to $(t, k, S, i) \in \fr \land t_i \leq t$. Therefore, we can further lower-bound \eqref{eq:welarge-parta} as
\begin{align}
    \eqref{eq:welarge-parta} ~\geq~ (1 - \epsilon^{1/4}) \cdot \sum_{t = 1}^{t_i} \sum_{k \in [K]} \sum_{S \ni i} x^*_{t, k, S} \cdot w_{t, k, S}(i) \cdot \one[(t, k, S, i) \in \freasy ] ~=~ (1 - \epsilon^{1/4}) \cdot \we(i), \label{eq:parta-final}
\end{align}
where the last equality follows from the definition of $\we$ and the observation that $\one[(t, k, S_i) \in \freasy] = \one[(t, k, S_i) \in \fr \land t \leq t_i]$.

\paragraph{Bounding \eqref{eq:welarge-partb}.} Recall that \eqref{eq:welarge-partb} represents the deterministic value we collect from agent $t_i$. Since the algorithm only picks free values from agent $t \in [1, t_i - 1]$, item $i$ should be available for agent $t_i$ if no free value arrives, i.e., we have $\pr[i \text{ is available for }t_i] \geq 1 - \epsilon^{1/4}$. Then,
\[
\eqref{eq:welarge-partb} ~\geq~ (1 - \epsilon^{1/4}) \cdot \sum_{k \in [K]} \sum_{S \ni i} x^*_{t_i, k, S} \cdot w_{t_i, k, S}(i)  \cdot \one[(t_i, k, S, i) \in \dt].
\]
To further bound the above inequality, we give the following \Cref{clm:dt-high}

\begin{Claim}
    \label{clm:dt-high}
    Fix $i$. Let $\dt'$ be a subset of $\dt$ that satisfies the following: For every $(t, k, S, i') \in \dt'$, there must be $i' = i$. Then,
    \[
     \sum_{(t, k, S, i) \in \dt'} x^*_{t, k, S} \cdot w_{t, k, S}(i) ~\geq~ \Big(0.5\cdot \sum_{(t, k, S, i) \in \dt'} x^*_{t, k, S} - 4\epsilon^{1/4}\Big) \cdot \lp_i(x^*)
    \]
\end{Claim}

\begin{proof}
     For simplicity of the notation, define $q = \sum_{(t, k, S, i) \in \dt'} x^*_{t, k, S}$. Then,
     \begin{align*}
         \lp_i(x^*) ~=~& \sum_{t, k, S} x^*_{t, k, S} \cdot  w_{t, k, S}(i) \cdot \one[(t, k, S, i) \in \dt'] + \sum_{t, k, S} x^*_{t, k, S} \cdot  w_{t, k, S}(i) \cdot \one[(t, k, S, i) \in \fr] \\
         +& \sum_{t, k, S} x^*_{t, k, S} \cdot  w_{t, k, S}(i) \cdot \one[(t, k, S, i) \in \dt \setminus \dt'] \\
         ~\leq~& \sum_{(t, k, S, i) \in \dt'} x^*_{t, k, S} \cdot w_{t, k, S}(i) + (0.5 + 2\epsilon^{1/4}) \cdot \lp_i(x^*) \\
         +& \sum_{t, k, S} x^*_{t, k, S} \cdot (0.5 + 2\epsilon^{1/4}) \cdot \lp_i(x^*) \cdot \one[(t, k, S, i) \notin \dt'] \\
         ~=~& \sum_{(t, k, S, i) \in \dt'} x^*_{t, k, S} \cdot w_{t, k, S}(i) + (0.5 + 2\epsilon^{1/4}) \cdot \left(\lp_i(x^*) + (1 -q)\cdot \lp_i(x^*) \right),
     \end{align*}
     where the inequality follows from the conditions in \Cref{lma:property1} for $i \in U$. Rearrange the inequality, we have
     \begin{align*}
         \sum_{(t, k, S, i) \in \dt'} x^*_{t, k, S} \cdot w_{t, k, S}(i) ~\geq~ (0.5q - 4\epsilon^{1/4}) \cdot \lp_i(x^*)
     \end{align*}
    which gives \Cref{clm:dt-high} after plugging the definition of $q$ back into the above expression.
\end{proof}

For simplicity of the notation, define
\[
q_1(i) ~:=~ \sum_{k \in [K]} \sum_{S \ni i} x^*_{t_i, k, S} \cdot \one[(t_i, k, S, i) \in \dt]
\]
to be the total probability mass of deterministic values from agent $t_i$. 
Applying \Cref{clm:dt-high} with $\dt' = \{(t', k', S', i') \in \dt: t' = t_i \land i' = i\}$, we have
\begin{align}
    \eqref{eq:welarge-partb} ~\geq~ (1 - \epsilon^{1/4}) \cdot (0.5q_1(i) -4\epsilon^{1/4}) \cdot \lp_i(x^*) ~\geq~ (0.5q_1(i) - 4.5\epsilon^{1/4}) \cdot \lp_i(x^*),\label{eq:partb-final}
\end{align}
where we use the fact that $q_1(i) \leq 1$ in the last inequality.

\paragraph{Bounding \eqref{eq:welarge-partc}.} 
We first bound $\pr[i \text{ is available for }t]$ for $t \in [t_i + 1, T]$. Recall that \eqref{eq:welarge-partc} represents the free value we collect from agent $t \in [t_i + 1, T]$. 
 Then, $i$ must be available for $t$ if the following three events do not happen:
\begin{itemize}
    \item Value $w_{t, k, S}(i)$ that satisfies $(t, k, S, i) \in \fr$ is sent to $\bssipi_i$.
    \item Value $w_{t_i, k, S}(i)$ that satisfies $(t_i, k, S, i) \in \dt$ is sent to $\bssipi_i$.
    \item Value $w_{t,k, S}(i)$ that satisfies $(t, k, S, i) \in \dt$ and $t \in [t_i + 1, T]$  is sent to $\bssipi_i$.
\end{itemize}

For simplicity of the notation, define
\[
q_2(i) ~:=~ \sum_{t = t_i + 1}^T \sum_{k \in [K]} \sum_{S \ni i} x^*_{t, k, S} \cdot \one[(t, k, S, i) \in \dt]
\]
to be the total probability mass of deterministic values from agent $t \in [t_i + 1, T]$. By the union bound and \Cref{lma:property1}, the first event happens with probability at most $\epsilon^{1/4}$. By the union bound and the definitions of $q_1(i)$ and $q_2(i)$, the second event happens with probability at most $q_1(i)$, and the third event happens with probability at most $q_2(i)$. Therefore, we have 
\[
\pr[i \text{ is available for }t] ~\geq~ 1 - \epsilon^{1/4} - q_1(i) - q_2(i).
\]
Plugging the above inequality into \eqref{eq:welarge-partc}, we have
\begin{align}
    \eqref{eq:welarge-partc} ~\geq~& (1 - \epsilon^{1/4} - q_1(i) - q_2(i)) \cdot \sum_{t = t_i + 1}^T \sum_{k \in [K]} \sum_{S \ni i} x^*_{t, k, S} \cdot w_{t, k, S}(i)  \cdot \one[(t, k, S, i) \in \fr] \notag \\
    ~=~& (1 - \epsilon^{1/4} - q_1(i) - q_2(i)) \cdot \sum_{t = 1}^T \sum_{k \in [K]} \sum_{S \ni i} x^*_{t, k, S} \cdot w_{t, k, S}(i)  \cdot \one[(t, k, S, i) \in \fr] \notag \\
    &~- (1 - \epsilon^{1/4} - q_1(i) - q_2(i)) \cdot \sum_{t = 1}^{t_i} \sum_{k \in [K]} \sum_{S \ni i} x^*_{t, k, S} \cdot w_{t, k, S}(i)  \cdot \one[(t, k, S, i) \in \fr] \notag \\
    ~\geq~& (1 - \epsilon^{1/4} - q_1(i) - q_2(i)) \cdot \left((0.5 - 2\epsilon^{1/4}) \cdot \lp_i(x^*) - \we(i) \right) \notag \\
    ~\geq~& (0.5 - 2.5\epsilon^{1/4} - 0.5q_1(i) - 0.5 q_2(i)) \cdot \lp_i(x^*) - (1 - q_1(i) - q_2(i)) \cdot \we(i),  \label{eq:partc-final}
\end{align}
where the first inequality uses \Cref{lma:property1} and the definition of $\we(i)$.

\paragraph{Bounding \eqref{eq:welarge-partd}.} Recall that \eqref{eq:welarge-partd} represents the deterministic value we collect from agent $t \in [t_i + 1, T]$. Therefore, similar to the proof for \eqref{eq:welarge-partc}, there must be
\[
\pr[i \text{ is available for } t] ~\geq~ 1 - \epsilon^{1/4} - q_1(i) - q_2(i)
\]
for $t \in [t_i + 1, T]$. Applying the above inequality to \eqref{eq:welarge-partd}, we have
\[
\eqref{eq:welarge-partd} ~\geq~ (1 - \epsilon^{1/4} - q_1(i) - q_2(i)) \cdot\sum_{t = t_i + 1}^T \sum_{k \in [K]} \sum_{S \ni i} x^*_{t, k, S} \cdot w_{t, k, S}(i)  \cdot \one[(t, k, S, i) \in \dt].
\]
 Next,  applying \Cref{clm:dt-high} with $\dt' = \{(t', k', S', i') \in \dt : t' \geq t_i + 1 \land i' = i\}$ to the above inequality gives
\begin{align}
\eqref{eq:welarge-partd} ~&\geq~ (1 - \epsilon^{1/4} - q_1(i) - q_2(i)) \cdot \left(0.5 q_2(i) - 4\epsilon^{1/4} \right) \cdot \lp_i(x^*). \notag \\
~&\geq~ \left(0.5 q_2(i) - 5 \epsilon^{1/4} - 0.5q_2(i) \cdot (q_1(i) + q_2(i))\right) \cdot \lp_i(x^*).\label{eq:partd-final}
\end{align}

\paragraph{Putting everything together.} Finally, we sum the bounds for \eqref{eq:welarge-parta}, \eqref{eq:welarge-partb}, \eqref{eq:welarge-partc}, and \eqref{eq:welarge-partd} together. We will further apply the following two inequalities to simplify the calculation:
\[
q_2(i) ~\leq~ \epseh \quad \text{and} \quad  \epseh ~\leq~ q_1(i) + q_2(i) ~\leq~ 1,
\]
where both inequalities follow from the definition of $q_1(i)$, $q_2(i)$, and $t_i$.

Summing \eqref{eq:parta-final}, \eqref{eq:partb-final}, \eqref{eq:partc-final}, and \eqref{eq:partd-final} together gives
\begin{align*}
    \opt(\bssipi_i) ~\geq~& (0.5 - 12\epsilon^{1/4}) \cdot \lp_i(x^*)  - \epsilon^{1/4} \cdot \we(i) \\
    &+~ (q_1(i) + q_2(i)) \cdot (\we(i) - 0.5 q_2(i) \cdot \lp_i(x^*)) \\
    ~\geq~& (0.5 - 13\epsilon^{1/4}) \cdot \lp_i(x^*) + (q_1(i) + q_2(i)) \cdot (\we(i) - 0.5\epseh \cdot \lp_i(x^*)) \\
    ~\geq~& (0.5 - 13\epsilon^{1/4}) \cdot \lp_i(x^*) + \epseh \cdot (\we(i) - 0.5\epseh \cdot \lp_i(x^*)),
\end{align*}
where the second inequality upper-bounds $q_2(i)$ by $\epseh$, and the last inequality lower-bounds $q_1(i) + q_2(i)$ by $\epseh$, which is applicable here because we assume $\we(i) > 0.5\epseh \cdot \lp_i(x^*)$.

Recall that the above bound for $\opt(\bssipi_i)$ assumes $\we(i) > 0.5 \epseh \cdot \lp_i(x^*)$, and we simply lower-bound $\opt(\bssipi_i)$ as $0.5 \lp_i(x^*)$ when $\we(i) \leq 0.5 \epseh \cdot \lp_i(x^*)$, i.e., we have
\begin{align*}
    \opt(\bssipi_i) ~\geq~ \begin{cases} 0.5 \lp_i(x^*) & \text{if }\we(i) \leq 0.5 \epseh \cdot \lp_i(x^*),  \\ (0.5 - 13\epsilon^{1/4} - 0.5 \epseh^2) \cdot \lp_i(x^*) + \epseh \cdot \we(i)  & \text{if }\we(i) > 0.5 \epseh \cdot \lp_i(x^*).
\end{cases}~
\end{align*}
Combining the above two cases, we can generally lower-bound $\opt(\bssipi_i)$ as
\[
\opt(\bssipi_i) ~\geq~ (0.5 - 13\epsilon^{1/4} - 0.5 \epseh^2) \cdot \lp_i(x^*) + \one[\we(i) > 0.5\epseh\cdot \lp_i(x^*)] \cdot \epseh \cdot \we(i).
\]
Summing the above inequality for $i \in [m]$, we have
\begin{align*}
    \sum_{i \in [m]} \opt(\bssipi_i) ~\geq~& (0.5 - 13\epsilon^{1/4} - 0.5 \epseh^2) \cdot \lp(x^*) + \sum_{i \in [m]} \epseh \cdot \we(i) \\
    &-~ \sum_{i \in [m]} \one[\we(i) \leq 0.5 \epseh \cdot \lp_i(x^*)] \cdot \epseh \cdot \we(i) \\
    ~\geq~& (0.5 - 13\epsilon^{1/4} - 0.5 \epseh^2) \cdot \lp(x^*) + \epseh \cdot \we - \epseh \cdot \sum_{i \in [m]} 0.5 \epseh \cdot \lp_i(x^*) \\
    ~\geq~& (0.5 - 13\epsilon^{1/4} - \epseh^2) \cdot \lp(x^*) + \epseh \cdot \we. \qedhere
\end{align*}
\end{proof}

\section{Omitted Proofs in \Cref{sec:prelim}}
\label{sec:prelim-appendix}

\subsection{Proof of \Cref{lma:polytime-solvable}}

\polysolve*

\begin{proof}
    For convenience, we restate the linear program $\lpon$ below (note the online constraint \eqref{eq:online-constraint} is rearranged):
\begin{align}
    \text{maximize } \quad \quad  &    \quad\sum_{t \in [T]}  \sum_{k \in [K]} \sum_{S \subseteq [m]}~ v_{t,k}(S) \cdot x_{t,k, S}, \tag{$\lpon$}  \\
    \text{s.t.}\qquad \forall i \in [m], &  \quad   \sum_{t \in [T]} \sum_{k \in [K]} \sum_{S \ni i}~  x_{t,k,S}  ~\leq~ 1  \notag  \\
    \forall t \in [T], k \in [K], &   \quad \sum_{S \subseteq [m]}~ x_{t,k, S} ~\leq~ p_{t, k} \notag\\
    \forall t \in [T],  k \in [K], i \in [m],&  \quad  \sum_{t' < t} \sum_{k' \in [K]} \sum_{S \ni i} p_{t,k} \cdot x_{t',k',S} + \sum_{S \ni i} x_{t, k, S} ~\leq~ p_{t,k} \notag \\
    \forall t \in [T],  k \in [K], S \subseteq [m],&  \quad  0 ~\leq~ x_{t, k, S} ~\leq~ 1. \notag
\end{align}
Let $\alpha_i, \beta_{t,k}, \gamma_{t,k,i}$ be the dual variables of the above LP. Then, the dual of $\lpon$ is as follows:
\begin{align}
    \text{minimize } \quad \quad  &    \quad\sum_{t \in [T]}  \sum_{k \in [K]} \sum_{S \subseteq [m]}~ v_{t,k}(S) \cdot x_{t,k, S}, \tag{$\lpdual$}  \\
    \text{s.t.}\qquad \forall t \in [T],  k \in [K], S \subseteq [m], &  \quad   \beta_{t,k}+ \sum_{i \in S} \left(\alpha_i + \gamma_{t,k,i} + \sum_{t' > t} \sum_{k' \in [K]}  p_{t', k'} \cdot \gamma_{t,k,i} \right)  ~\leq~ v_{t,k}(S)  \notag  \\
    \forall i \in [m], &  \quad   \alpha_i \geq 0  \notag \\
     \forall t \in [T], k  \in [K], &  \quad   \beta_{t,k} \geq 0  \notag \\
      \forall t \in [T], k  \in [K], i \in [m], &  \quad   \gamma_{t,k,i} \geq 0  \notag 
\end{align}
Note that a demand oracle access to every function $v_{t, k}(S)$ gives a separation oracle to $\lpdual$: Given $\{\alpha_i\}, \{\beta_{t,k}\}, \{\gamma_{t,k,i}\}$, we set 
\[
\rho_i = \alpha_i + \gamma_{t,k,i} + \sum_{t' > t} \sum_{k' \in [K]}  p_{t', k'} \cdot \gamma_{t,k,i},
\]
and call the demand oracle to find subset $S_{t,k}$, such that $v_{t,k}(S_{t,k}) - \sum_{i \in S_{t,k}} \rho_i$ is maximized. Then, the given solution of $\lpdual$ is feasible if and only if $v_{t,k}(S_{t,k}) - \sum_{i \in S_{t,k}} \rho_i \geq \beta_{t,k}$ is satisfied for every $t \in [T], k \in [K]$. By using the Ellipsoid method, $\lpdual$ is solvable  in $\poly(T,k,m)$ time.

Next, we show the optimal solution $\{x^*_{t,k, S}\}$ of $\lpon$ has $\poly(T,K, m)$ non-zero entries. By complementary slackness, it's sufficient to assume that $x^*_{t,k,S} > 0$ only when the constraint corresponding to the variable $x^*_{t,k, S}$ is encountered when solving $\lpdual$. Since the Ellipsoid method finishes in $\poly(T,k,m)$ time, at most $\poly(T,k,m)$ constraints in $\lpdual$ are encountered. Implying that $\{x^*_{t,k, S}\}$ has $\poly(T,K, m)$ non-zero entries.
\end{proof}

\subsection{Properties of Submodular Functions}

\restrict*

\begin{proof}
    The non-negativity of $f(B \mid A)$ follows from the fact that $A \subseteq B \cup A$, and function $f$ is monotone. The monotonicity of $f(B \mid A)$ follows from the fact that for any $B \subseteq B' \subseteq [m]$, there must be $B \cup A \subseteq B' \cup A$. Then, $f(B' \mid A) - f(B \mid A) = f(B' \cup A) - f(B \cup A) \geq 0$. 

    For the submodularity of $f(B \mid A)$, Let $B \subseteq B'$ and $i \in [m]$. We aims at showing
    \[
    f(B \cup \{i\} \mid A) - f(B \mid A) ~\geq~ f(B' \cup \{i\} \mid A) - f(B' \mid A),
    \]
    which is equivalent to showing
    \[
    f(B \cup \{i\} \cup A) + f(B' \cup A) ~\geq~ f(B' \cup \{i\} \cup A) + f(B \cup A).
    \]
    Let $U = B \cup \{i\} \cup A$ and $V = B' \cup A$. Then, the above inequality holds from the submodularity of function $f$ and the observation that $U \cap V = B \cup A$, while $U \cup B = B' \cup \{i\} A$.

    Finally, we show $f(B \mid A) \geq f(B \mid A')$ holds when $A \subseteq A'$. Note that this is equivalent to proving 
    \[
    f(B \cup A) + f(A') ~\geq~ f(B \cup A') + f(A).
    \]
    Note that $B \cup A \cup A' = B \cup A'$, and $A \subseteq (B \cup A) \cap A'$. Then, the submodularity and monotonicity of $f$ gives
    \[
    f(B \cup A) + f(A') ~\geq~ f(B \cup A') + f((B \cup A) \cap A') ~\geq~ f(B \cup A') + f(A). \qedhere
    \]
\end{proof}

\additive*

\begin{proof}
    We have
    \begin{align*}
        \E_{S \sim D^{(A)}}[f(S)] ~&=~ \sum_{i \in [m]} \E_{S \sim D^{(A)}} [f(S \cap [i]) - f(S \cap [i-1])] \\
        ~&\geq~ \sum_{i \in [m]} \E_{S \sim D^{(A)}} [f(A \cap [i]) - f(A \cap [i-1])] \\
        ~&=~ \sum_{i \in A} \big(f(A \cap [i]) - f(A \cap [i-1]) \big) \cdot \pr_{ S \sim D^{(A)}}[i \in S] \\
        ~&\geq~ p \cdot \sum_{i \in A} \big(f(A \cap [i]) - f(A \cap [i-1]) \big) ~=~ p \cdot f(A),
    \end{align*}
    where we use the submodularity of $f$ and the fact that $S \subseteq A$ for $S \sim D^{(A)}$ in the first inequality.
\end{proof}

\telescope*

\begin{proof}
    We have
    \begin{align*}
        f(A) ~=~ \sum_{i \in A} \big(f(A \cap [i]) - f(A \cap [i-1]) \big) ~\geq~ \sum_{i \in A} \big(f(B \cap [i]) - f(B \cap [i-1]) \big),
    \end{align*}
    where the inequality follows from the submodularity of $f$ and the assumption that $A \subseteq B$.
\end{proof}

\section{Omitted Proofs in \Cref{sec:submod}}
\label{sec:submod-appendix}

In this section, we provide the missing proofs of 
 \Cref{lma:property1} and \Cref{lma:free_prob}.

\property*

\freeprob*

As the proof of \Cref{lma:property1} relies on \Cref{lma:free_prob}, we first prove \Cref{lma:free_prob}:

\begin{proof}[Proof of \Cref{lma:free_prob}]
Recall that \cite{stoc/STW25} already provided a proof of \Cref{lma:free_prob} assuming the prophet inequality instance $\sipi$ is Bernoulli, and no fixed threshold algorithm, i.e., an algorithm that picks the first arriving value greater than some fixed parameter $\tau$, can achieve an expected reward of more than $(0.5 + \mu) \cdot \bm(\sipi)$. Therefore, to prove \Cref{lma:free_prob}, we provide a reduction from a non-Bernoulli instance to a Bernoulli instance. To be specific, let $\sipi_i$ be a non-Bernoulli instance that satisfies $\opt(\sipi) \leq (0.5 + \mu) \cdot \bm(\sipi)$. Let $\D^{(w)}_t(\sipi)$ be the $t$-th underlying distribution of $\sipi$. Assume $\D^{(w)}_t(\sipi)$ generates non-zero value $w_{t,k}$ with probability $p_{t,k}$. Without loss of generality, we assume that every $w_{t,k}$ is unique (which can be achieved by adding an arbitrarily small noise to each value).  Then, there exists a Bernoulli instance $\sipi'$ that satisfies the following properties:
\begin{itemize}
    \item Assume the underlying distributions of $\sipi'$ generates non-zero value $w'_{t,k}$ with probability $p'_{t,k}$. Then, the tuple set $\{(w_{t,k}, p_{t,k}\}$ is identical to the set $\{(w'_{t,k}, p'_{t,k})\}$. Note that this guarantees $\bm(\sipi) = \bm(\sipi')$, and the two statements in \Cref{lma:free_prob} are identical for $\sipi$ and $\sipi'$.
    \item No fixed threshold algorithm can achieves an expected value of at least $(0.5 + \mu) \cdot \bm(\sipi)$.
\end{itemize}

For instance $\sipi$, we call it to be a ``good'' instance, if no ``fixed subset'' algorithm can achieve an expected reward of more than $(0.5 + \mu) \cdot \bm(\sipi)$. Here, we define a ``fixed subset'' algorithm as follows: Fix a subset $S$ of $\{(w_{t,k}, p_{t,k}\}$. The algorithm only picks the value $w_t$ if  the arriving value $w_t$ with type $k$ satisfies $(w_{t,k}, p_{t,k}) \in S$. Then, it's sufficient to show that every good instance $\sipi$ has a corresponding Bernoulli instance $\sipi'$ that satisfies the conditions above. 

We prove via introduction. For instance $\sipi$, let $K_t$ be the number of different non-zero values that distribution $\D^{(w)}_t(\sipi)$ may generate. Define
\[
d(\sipi) ~:=~\sum_{t \in [T]} K_t - 1
\]
to be the ``degree of Bernoulli'' of instance $\sipi$, i.e. when $d(\sipi) = 0$, the instance is purely Bernoulli; when $d(\sipi)$ becomes large, the instance is becoming more and more non-Bernoulli.

The base case of the induction is $d = 0$. The argument trivially holds, as one can observe the class of ``fixed threshold'' algorithms is a subset of the class of ``fixed subset'' algorithms.

Now suppose the desired argument holds for every good instance $\sipi$ that satisfies $d(\sipi) \leq D$. Now consider a good instance $\sipi$ that satisfies $d(\sipi) = D + 1$. The proof idea is to construct a new instance $\widehat{\sipi}$ that does not change the tuple set $\{(w_{t,k}, p_{t,k})\}$, but with $d(\widehat{\sipi}) = D$. Then, the induction hypothesis guarantees that the desired Bernoulli instance $\sipi'$ for $\widehat{\sipi}$ also works for $\sipi$.

The instance $\widehat{\sipi}$ is constructed as follows: Let $w_{t^*,k^*}$ be the minimum positive value that satisfies $K_{t^*} \geq 2$, i.e., $w_{t^*,k^*}$ belongs to a non-Bernoulli distribution. Now consider to remove value $w_{t^*,k^*}$ from distribution $\D^{(w)}_{t^*}(\sipi)$, and add a new Bernoulli distribution in front of $t^*$ that generates $w_{t^*,k^*}$ with probability $p_{t^*,k^*}$ and $0$ otherwise. Note that this construction does not change the tuple set $\{(w_{t,k}, p_{t,k})\}$, but decreases the value of $K_{t^*}$ by $1$, i.e., we have $d(\widehat{\sipi}) = d(\sipi) - 1 = D$.

It remains to show that $\widehat{\sipi}$ is also a good instance. We prove via contradiction: Suppose running a fixed subset algorithm with set $S \subseteq \{(w_{t,k}, p_{t,k})\}$ on instance $\widehat{\sipi}$ achieves an expected reward strictly greater than $(0.5 + \mu) \cdot \bm(\sipi)$ (note that we have $\bm(\sipi) = \bm(\widehat{\sipi})$). Then, there must be $(w_{t^*, k^*}, p_{t^*, k^*}) \in S$, otherwise the algorithm performs identically on $\sipi$ and $\widehat{\sipi}$. For simplicity of the notation, we define $p_1 = p_{t^*, k^*}$, $w_1 = w_{t^*, k^*}$. Next, we define $p_2$ to be the probability that $\D^{(w)}_t(\sipi)$ generates a non-zero value that does not equal to $w_{t^*, k^*}$ and should be taken by the fixed subset algorithm, and $w_2$ to be the expectation of values generated from $\D^{(w)}_t(\sipi)$ conditioning on that the value does not equal to $w_{t^*, k^*}$ and should be taken by the fixed subset algorithm. Finally, we define $p_0$ to be the probability that the fixed subset algorithm does not take any value from $[1, t-1]$, and let $w_3$ be the expected outcome from agents $[t+1, T]$.  We note that $p_0$ and $w_3$ are both identical for $\sipi$ and $\widehat{\sipi}$. Then, if we run the fixed subset algorithm on $\sipi$, the expected outcome starting from $t$ is
\[
p_0 \cdot \left(p_1 \cdot w_1 + p_2 \cdot w_2 + (1 - p_1 - p_2) \cdot w_3\right).
\]
On the other hand, if we run the fixed subset algorithm on $\widehat{\sipi}$, the expected outcome starting from the newly inserted Bernoulli distribution is
\[
p_0 \cdot \left(p_1 \cdot w_1 + (1 - p_1) \cdot p_2 \cdot w_2 + (1 - p_1)(1-p_2) \cdot w_3\right).
\]
Since the fixed subset algorithm gains at most $(0.5 + \mu) \bm(\sipi)$ on instance $\sipi$, and more than $(0.5 + \mu) \bm(\sipi)$ on instance $\widehat{\sipi}$, the first quantity should be strictly smaller than the second quantity, which gives
\[
p_1 \cdot w_1 + p_2 \cdot w_2 + (1 - p_1 - p_2) \cdot w_3 < p_1 \cdot w_1 + (1 - p_1) \cdot p_2 \cdot w_2 + (1 - p_1)(1-p_2) \cdot w_3.
\]
Rearrange the above inequality, it implies $w_2 < w_3$.

Now, consider to run a fixed subset algorithm on $\sipi$ with $S' = S \setminus \{(w_{t,k}, p_{t,k}):t = t^*\}$. Then, the expected outcome starting from value $t$ is
\begin{align*}
    p_0 \cdot w_3 ~&\geq~ p_0 \cdot \left(w_3 - p_1(w_3 - w_1) - p_2(w_3 - w_2)\right) \\
    ~&=~ p_0 \cdot \left(p_1 \cdot w_1 + p_2 \cdot w_2 + (1 - p_1 - p_2) \cdot w_3\right),
\end{align*}
where we use the fact that $w_3 > w_2 \geq w_1$ in the first inequality (recall that we assume $w_1$ is the minimum non-zero value in the support of $\D^{(w)}_{t^*}(\sipi)$). Since we assume running a fixed subset algorithm with $S$ on $\widehat{\sipi}$ achieves an expected reward of more than $(0.5 + \mu) \cdot \bm(\sipi)$, the above inequality implies that running a fixed subset algorithm with $S'$ on $\sipi$ should achieve strictly greater than $(0.5 + \mu) \cdot \bm(\sipi)$, which is in contrast to the assumption that $\sipi$ is good. Therefore, $\widehat{\sipi}$ is also good.
\end{proof}

Now, we prove \Cref{lma:property1}:

\begin{proof}[Proof of \Cref{lma:property1}]
    Let $\bssipi_i$ be the revised single-item prophet inequality instance for item $i$ in \Cref{alg:submod-half}. Recall that $\bm(\bssipi_i) = \lp_i(x^*)$. Then, we define the desired subset $U \subseteq [m]$ for \Cref{lma:property1} as follows:
    \[
    U ~:=~ \{i \in [m]: \opt(\bssipi_i) ~\leq~ (0.5 + \epsilon^{3/4}/4) \cdot \lp_i(x^*)\}.
    \]
    We first show that the total weight in $[m] \setminus U$ can't be too much via contradiction. Suppose we have $\sum_{i \in [m] \setminus U} \lp_i(x^*) > 4\epsilon^{1/4} \cdot \lp(x^*)$. This implies 
    \begin{align*} 
       \sum_{i \in [m]} \opt(\bssipi_i) ~&=~ \sum_{i \in U} \opt(\bssipi_i) + \sum_{i \in [m] \setminus U} \opt(\bssipi_i) \\
        ~&\geq~ \sum_{i \in U} 0.5 \cdot \lp_i(x^*) + \sum_{i \in [m] \setminus U} (0.5 + \frac{1}{4} \cdot \epsilon^{3/4}) \cdot \lp_i(x^*) \\
        ~&=~ 0.5 \lp(x^*) + \frac{1}{4} \cdot \epsilon^{3/4} \cdot \sum_{i \in [m] \setminus U} \lp_i(x^*) \\
        ~&>~ (0.5 + \frac{1}{4} \cdot \epsilon^{3/4} \cdot 4\epsilon^{1/4}) \lp(x^*) ~=~ (0.5 + \epsilon) \cdot \lp(x^*),
    \end{align*}
    where the first inequality follows from \Cref{lma:reward-submod-to-si}, and the last inequality uses the above assumption. Since the above inequality is in contrast to the assumption we have in \Cref{lma:property1}. Therefore, there must be 
    \[
    \sum_{i \in [m] \setminus U} \lp_i(x^*) \leq 4\epsilon^{1/4} \cdot \lp(x^*).
    \]
    Next, we apply \Cref{lma:free_prob} with $\mu = \epsilon^{3/4}/4$ and $\beta = 0.5 + 2\epsilon^{1/4}$, and let  $\delta = \sqrt{4\mu/(\beta - 0.5 - \mu)}$. Then, for every $i \in U$, the first statement in \Cref{lma:free_prob} gives
    \begin{align*}
        \sum_{t \in [T]} \sum_{k, S} x^*_{t, k, S} \cdot \one[(t, k, S, i) \in \fr] ~=~ \sum_{t \in [T]} \sum_{k, S} x^*_{t, k, S} \cdot \one[w_{t, k, S}(i) > 2\epsilon^{1/4} \cdot \lp_i(x^*)] ~\leq~ \delta ~\leq~ \epsilon^{1/4},
    \end{align*}
    where the last inequality follows from the definition of $\delta$. Note that for every $i \notin U$, $\sum_{t, k, S} x^*_{t, k, S} \cdot \one[(t, k, S, i) \in \fr]$ is always $0$. Therefore, $\sum_{t, k, S} x^*_{t, k, S} \cdot \one[(t, k, S, i) \in \fr] \leq \epsilon^{1/4}$ holds for every $i \in [m]$.

    Now, we fix $i \in U$. The second statement in \Cref{lma:free_prob} gives
    \begin{align*}
        \sum_{t, k, S} x^*_{t, k, S} \cdot w_{t, k, S}(i) \cdot \one[(t, k, S, i) \in \fr] ~&\leq~ \frac{0.5 + \mu}{1 - \delta} \cdot \lp_i(x^*) \\
         ~&\leq~ (0.5 + \mu + \delta) \cdot \lp_i(x^*) ~\leq~ (0.5 + 2\epsilon^{1/4}) \cdot \lp_i(x^*),
    \end{align*}
    where the second inequality holds when $0.5 \geq \mu + \delta$, which is true when $\epsilon < 10^{-4}$, and the last inequality holds when $\epsilon < 0.1$.

    It remains to show $(0.5 - 2\epsilon^{1/4}) \cdot \lp_i(x^*) \leq \sum_{t, k, S} x^*_{t, k, S} \cdot w_{t, k, S}(i) \cdot \one[(t, k, S, i) \in \fr]$. Note that when $i \in U$, we have
    \begin{align*}
        \lp_i(x^*) ~&=~ \sum_{t, k, S} x^*_{t, k, S} \cdot w_{t, k, S}(i) \cdot (\one[(t, k, S, i) \in \dt] + \one[(t, k, S, i) \in \fr]) \\
        ~&\leq~ \sum_{t, k, S} x^*_{t, k, S} \cdot (0.5 + 2\epsilon^{1/4}) \cdot \lp_i(x^*) + \sum_{t, k, S} x^*_{t, k, S} \cdot w_{t, k, S}(i) \cdot \one[(t, k, S, i) \in \fr] \\
        ~&=~ (0.5 + 2\epsilon^{1/4}) \cdot \lp_i(x^*) + \sum_{t, k, S} x^*_{t, k, S} \cdot w_{t, k, S}(i) \cdot \one[(t, k, S, i) \in \fr],
    \end{align*}
    where the inequality follows from the fact that when $i \in U$, there must be $w_{t, k, S}(i) \leq (0.5 + 2\epsilon^{1/4}) \cdot \lp_i(x^*)$ for every $(t,k, S, i) \in \dt$. Rearranging the above inequality finishes the proof.
\end{proof}

\end{document}